\documentclass{article}
\usepackage{ijcai19}

\usepackage{times}
\usepackage{soul}
\usepackage{url}
\usepackage[hidelinks]{hyperref}
\usepackage[utf8]{inputenc}
\usepackage[small]{caption}
\usepackage{graphicx}
\usepackage{amsmath}
\usepackage{amsfonts}
\usepackage{booktabs}
\usepackage{algorithm}
\usepackage{algorithmic}
\usepackage{amsthm}
\usepackage{subfig}
\title{AsymDPOP: Complete Inference for Asymmetric Distributed Constraint Optimization Problems}

\author{
Yanchen Deng
\and
Ziyu Chen\footnote{Corresponding Author}\and
Dingding Chen\and
Wenxin Zhang\And
Xingqiong Jiang\\
\affiliations
College of Computer Science, Chongqing University\\
\emails
dyc941126@126.com,
\{chenziyu,dingding\}@cqu.edu.cn, 
wenxinzhang18@163.com,
jxq@cqu.edu.cn
}

\begin{document}

\maketitle

\begin{abstract}
  Asymmetric distributed constraint optimization problems (ADCOPs) are an emerging model for coordinating agents with personal preferences. However, the existing inference-based complete algorithms which use local eliminations cannot be applied to ADCOPs, as the parent agents are required to transfer their private functions to their children. Rather than disclosing private functions explicitly to facilitate local eliminations, we solve the problem by enforcing delayed eliminations and propose AsymDPOP, the first inference-based complete algorithm for ADCOPs. To solve the severe scalability problems incurred by delayed eliminations, we propose to reduce the memory consumption by propagating a set of smaller utility tables instead of a joint utility table, and to reduce the computation efforts by sequential optimizations instead of joint optimizations.  The empirical evaluation indicates that AsymDPOP significantly outperforms the state-of-the-art, as well as the vanilla DPOP with PEAV formulation.
\end{abstract} 

\section{Introduction}
Distributed constraint optimization problems (DCOPs) \cite{modi2005adopt,fiorettoPY18} are a fundamental framework in multi-agent systems where agents coordinate their decisions to optimize a global objective. DCOPs have been adopted to model many real world problems including radio frequency allocation \cite{monteiro2012multi}, smart grid \cite{fioretto2017distributed} and distributed scheduling \cite{maheswaranTBPV04,liNL16}. 

Most of complete algorithms for DCOPs employ either distributed search \cite{hirayamay97,modi2005adopt,gershman2009asynchronous,yeoh2010bnb} or inference \cite{petcuF05,vinyals2009generalizing} to optimally solve DCOPs. However, since DCOPs are NP-hard, complete algorithms cannot scale up due to exponential overheads. Thus, incomplete algorithms \cite{maheswaran2004distributed,Zhang2005Distributed,okamoto2016distributed,rogers2011bounded,zivan2012max,chenDWH18,ottens2017duct,fioretto2016dynamic} are proposed to trade optimality for smaller computational efforts. 

Unfortunately, DCOPs fail to capture the ubiquitous asymmetric structure \cite{burke2007supply,maheswaranTBPV04,ramchurn2011agent} since each constrained agent shares the same payoffs. PEAV \cite{maheswaranTBPV04} attempts to capture the asymmetric costs by introducing mirror variables and the consistency is enforced by hard constraints. However, PEAV suffers from scalability problems since the number of variables significantly increases. Moreover, many classical DCOP algorithms perform poorly when applied to the formulation due to the presence of hard constraints \cite{grinshpounGZNM13}. On the other side, ADCOPs \cite{grinshpounGZNM13} are another framework that captures asymmetry by explicitly defining the exact payoff for each participant of a constraint without introducing any variables, which has been intensively investigated in recent years.

Solving ADCOPs involves evaluating and aggregating the payoff for each constrained agent, which is challenging in asymmetric settings due to a privacy concern. SyncABB and ATWB \cite{grinshpounGZNM13} are asymmetric adaption of SyncBB \cite{hirayamay97} and AFB \cite{gershman2009asynchronous}, using an one-phase strategy to aggregate the individual costs. That is, the algorithms systematically check each side of a constraint before reaching a complete assignment. Besides, AsymPT-FB \cite{litov2017forward} is the first tree-based algorithm for ADCOPs, which uses forward bounding to compute lower bounds and back bounding to achieve one-phase check. Recently, PT-ISABB \cite{2019arXiv190206039D} was proposed to improve the tree-based search by implementing a non-local elimination version of ADPOP \cite{petcu2005approximations} to provide much tighter lower bounds. However, since it relies on an exhaustive search to guarantee the optimality, the algorithm still suffers from exponential communication overheads. On the other hand, although complete inference algorithms (e.g., DPOP \cite{petcuF05}) only require a linear number of messages to solve DCOPs, they cannot be directly applied to ADCOPs without PEAV due to their requirement for complete knowledge of each constraint to facilitate variable elimination. Accordingly, the parents have to transfer their private cost functions to their children, which leaks at least a half of privacy.

In this paper, we adapt DPOP for solving ADCOPs for the first time by deferring the eliminations of variables. Specifically, we contribute to the state-of-the-art in the following aspects.
\begin{itemize}
	\item We propose AsymDPOP, the first complete inference-based algorithm to solve ADCOPs, by generalizing non-local elimination \cite{2019arXiv190206039D}. That is, instead of eliminating variables at their parents, we postpone the eliminations until their highest neighbors in the pseudo tree. In other words, an agent in our algorithm may be responsible for eliminating several variables. 
	\item We theoretically analyze the complexity of our algorithm where the space complexity of an agent is not only exponential in the number of its separators but also the number of its non-eliminated descendants. 
	\item We scale up our algorithm by introducing a table-set propagation scheme to reduce the memory consumption and a mini-batch scheme to reduce the number of operations when performing eliminations. Our empirical evaluation indicates that our proposed algorithm significantly outperforms the state-of-the-art, as well as the vanilla DPOP with PEAV formulation.
\end{itemize}

\section{Backgrounds}
In this section we introduce the preliminaries including DCOPs, ADCOPs, pseudo tree, DPOP and non-local elimination.
\subsection{Distributed Constraint Optimization Problems}
A distributed constraint optimization problem \cite{modi2005adopt} is defined by a tuple $\langle A,X,D,F\rangle$ where
\begin{itemize}
	\item $A=\{a_1,\dots, a_n\}$ is the set of agents
	\item $X=\{x_1,\dots,x_m\}$ is the set of variables
	\item $D=\{D_1,\dots, D_m\}$ is the set of domains. Variable $x_i$ takes values from $D_i$
	\item $F=\{f_1,\dots,f_q\}$ is the set of constraint functions. Each function $f_i:D_{i1}\times\dots\times D_{ik}\rightarrow\mathbb{R}_{\ge 0}$ specifies the cost assigned to each combination of $x_{i1},\dots,x_{ik}$.
\end{itemize}

For the sake of simplicity, we assume that each agent controls a variable (and thus the term "agent" and "variable" could be used interchangeably) and all constraint functions are binary (i.e., $f_{ij}: D_i\times D_j\rightarrow \mathbb{R}_{\ge 0}$). A solution to a DCOP is an assignment to all the variables such that the total cost is minimized. That is,
$$
X^*=\mathop{\arg\min}_{d_i\in D_i, d_j\in D_j}\sum_{f_{ij}\in F}f_{ij}(x_i=d_i,x_j=d_j)
$$
\begin{figure}
	\centering
	\includegraphics[scale=.29]{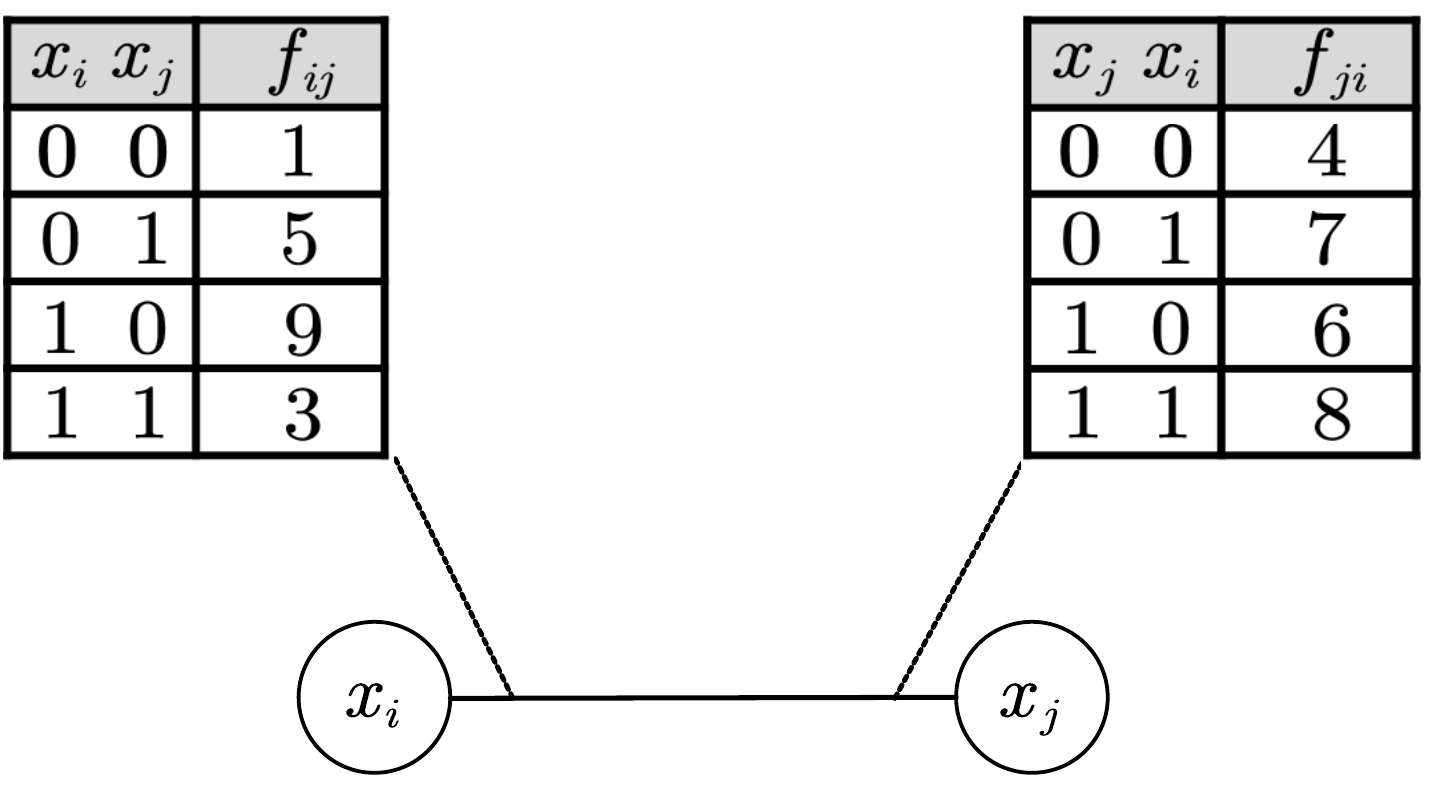}
	\caption{An ADCOP with two variables and a constraint}
	\label{adcop}
\end{figure}
\subsection{Asymmetric Distributed Constraint Optimization Problems}
While DCOPs assume an equal payoff for each participant of each constraint, asymmetric distributed constraint optimization problems (ADCOPs) \cite{grinshpounGZNM13} explicitly define the exact payoff for each constrained agent. In other words, a constraint function $f_i:D_{i1}\times\dots\times D_{ik}\rightarrow\mathbb{R}^k_{\ge 0}$ in an ADCOP specifies a cost vector for each possible combination of involved variables. And the goal is to find a solution which minimizes the aggregated cost. An ADCOP can be visualized by a constraint graph where the vertexes denote variables and the edges denote constraints. Fig. \ref{adcop} presents an ADCOP with two variables and a constraint.  Besides, for the constraint between $x_i$ and $x_j$, we denote the private function for $x_i$ and $x_j$ as $f_{ij}$ and $f_{ji}$, respectively.

\subsection{Pseudo Tree}
A pseudo tree \cite{freuderQ85} is an ordered arrangement to a constraint graph in which different branches are independent. A pseudo tree can be generated by a depth-first traverse to a constraint graph, categorizing constraints into tree edges and pseudo edges (i.e., non-tree edges). The neighbors of an agent $a_i$ are therefore categorized into its parent $P(a_i)$, pseudo parents $PP(a_i)$, children $C(a_i)$ and pseudo children $PC(a_i)$ according to their positions in the pseudo tree and the types of edges they connect through. We also denote its parent and pseudo parents as $AP(a_i)=P(a_i)\cup PP(a_i)$, and its descendants as $Desc(a_i)$. Besides, we denote its separators, i.e., the set of ancestors which are constrained with $a_i$	and its descendants, as $Sep(a_i)$ \cite{petcu2006odpop}. Finally, we denote $a_i$'s interface descendants, the set of descendants which are constrained with $Sep(a_i)$, as $ID(a_i)$. Fig. \ref{pseudo-tree} presents a pseudo tree in which the dotted edge is a pseudo edge.
\begin{figure}
	\centering
	\includegraphics[scale=.12]{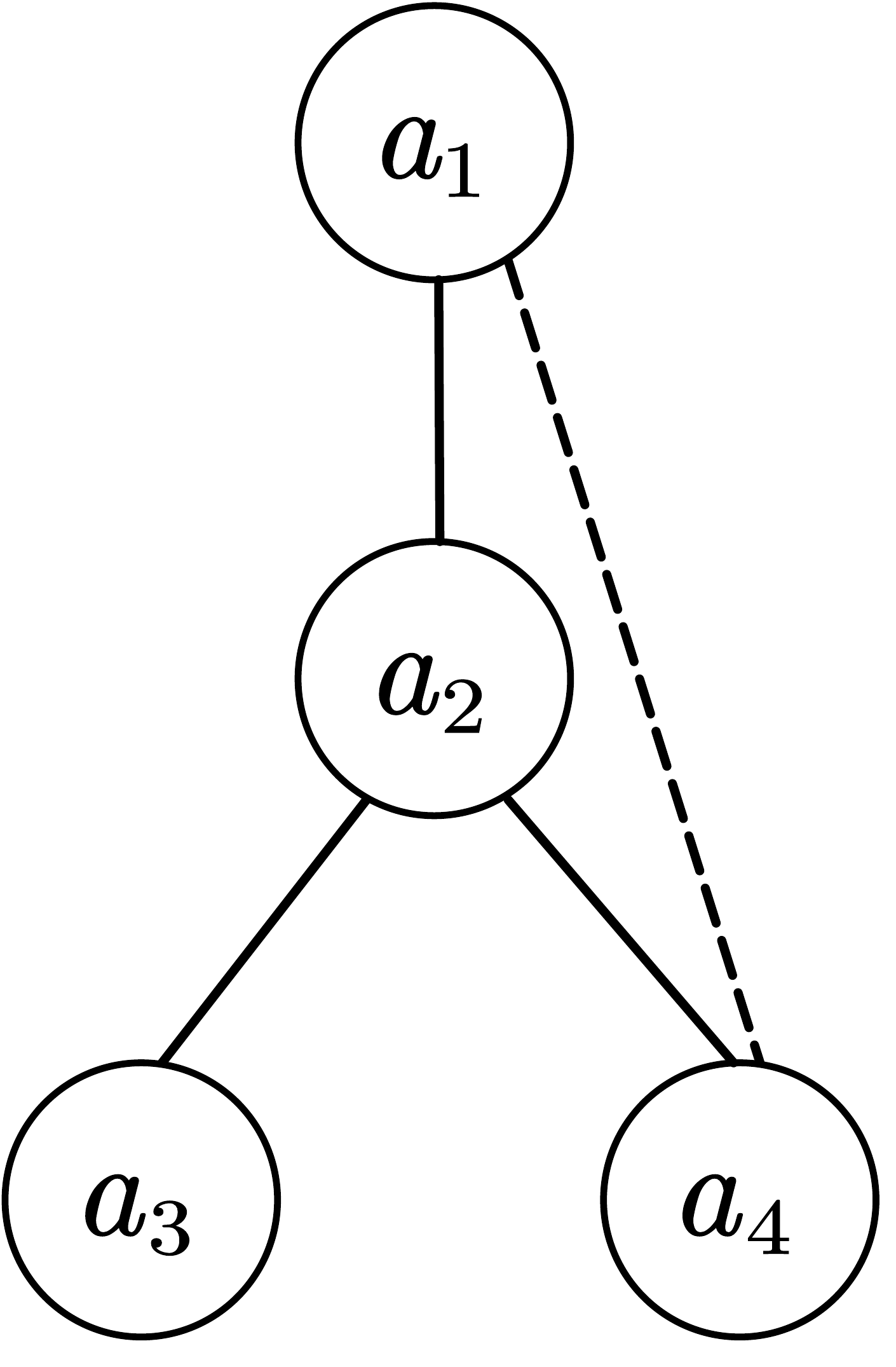}
	\caption{A pseudo tree}
	\label{pseudo-tree}
\end{figure}
\subsection{DPOP and Non-local Elimination}
DPOP \cite{petcuF05} is an inference-based complete algorithm for DCOPs based on bucket elimination \cite{dechter1999bucket}. Given a pseudo tree, it performs a bottom-up utility propagation phase to eliminate variables and a value propagation phase to assign the optimal assignment for each variable. More specifically, in the utility propagation phase, an agent $a_i$ eliminates its variables from the joint utility table by computing the optimal utility for each possible assignment to $Sep(a_i)$ after receiving the utility tables from its children, and sends the projected utility table to its parent. In the value propagation phase, $a_i$ computes the optimal assignments for its variables by considering the assignments received from its parent, and propagates the joint assignment to its children. Although DPOP only requires a linear number of messages to solve a DCOP, its memory consumption is exponential in the induced width. Thus, several tradeoffs including ADPOP \cite{petcu2005approximations}, MB-DPOP \cite{petcu2007mb} and ODPOP \cite{petcu2006odpop} have been proposed to improve its scalability.

However, DPOP cannot be directly applied to asymmetric settings as it requires the total knowledge of each constraint to perform optimal elimination locally. PT-ISABB \cite{2019arXiv190206039D} applies (A)DPOP into solving ADCOPs by performing variable elimination only to a subset of constraints to build look-up tables for lower bounds, and uses a tree-based search to guarantee the optimality. The algorithm further reinforces the bounds by a non-local elimination scheme. That is, instead of performing elimination locally, the elimination of a variable is postponed to its parent to include the private function enforced in the parent's side and increase the integrity of the utility table. 

\section{Asymmetric DPOP}
The existing complete algorithms for ADCOPs use complete search to exhaust the search space, which makes them unsuitable for large scale applications. In fact, as shown in our experimental results, these search-based algorithms can only solve the problems with the agent number less than 20. Hence, in this section, we propose a complete, privacy-protecting, and scalable inference-based algorithm for ADCOPs built upon generalized non-local elimination, called AsymDPOP. An execution example can be found in the appendix (https://arxiv.org/abs/1905.11828).

\begin{figure}
	\centering
	\includegraphics[scale=0.85]{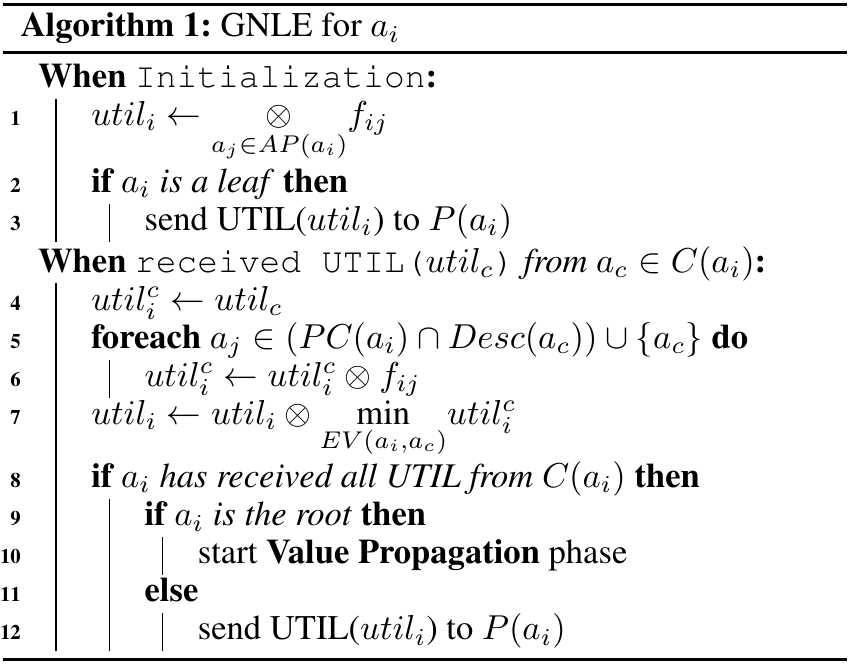}
\end{figure}
\subsection{Utility Propagation Phase}
In DPOP, a variable could be eliminated locally without loss of completeness after receiving all the utility messages from its children since all functions that involve the variable have been aggregated. However, the conclusion does not hold for ADCOPs. Taking Fig. \ref{pseudo-tree} as an example, $x_4$ cannot be eliminated locally since the private functions $f_{14}$ and $f_{24}$ are not given. Thus, the local elimination to $x_4$ w.r.t. $f_{41}$ and $f_{42}$ would lead to overestimate bias and offer no guarantee on the completeness. A na\"ive solution would be that $x_1$ and $x_2$ transfer their private functions to their children, which would lead to an unacceptable privacy loss.

Inspired by non-local elimination, we consider an alternative privacy-protecting approach to aggregate constraint functions. That is, instead of deferring to its parent, we postpone the elimination of a variable to its highest (pseudo) parent. In this way, all the functions involving the variables have been aggregated and the variable can be eliminated from the utility table optimally. Note that the utility table is a summation of the utility tables from children and local constraints. As a result, although the utility table which contains the variable's private functions is propagated to its ancestors without elimination, the ancestors can hardly infer the exact payoffs in these private functions. We refer the bottom-up utility propagation phase as generalized non-local elimination (GNLE).

Before diving into the details of GNLE, let's first consider the following definitions.
\newtheorem{definition}{Definition}
\begin{definition}[$dims$]
	The $dims(\cdot)$ is a function which returns the set of dimensions of a utility table.
\end{definition}
\begin{definition}[Slice]
	Let $S$ be a set of key-value pairs. $S_{[K]}$ is a slice of $S$ over $K$ such that
	$$S_{[K]}=\{(k=v)\in S|k\in K\}$$
\end{definition}
\begin{definition}[Join \cite{vinyals2009generalizing}]
	Let $U,U^\prime$ be two utility tables and $D_U=\times_{x_i\in dims(U)}D_i,D_{U^\prime}=\times_{x_i\in dims(U^\prime)}D_i$ be their joint domain spaces. $U\otimes U^\prime$ is the join of $U$ and $U^\prime$ over $D_{U\otimes U^\prime}=\times_{x_i\in dims(U)\cup dims(U^\prime)}D_i$ such that
	$$(U\otimes U^\prime)(V)=U(V_{[dims(U)]})+U^\prime(V_{[dims(U^\prime)]}),\forall V\in D_{U\otimes U^\prime}$$
\end{definition}

Algorithm 1 presents the sketch of GNLE. The algorithm begins with leaf agents sending their utility tables to their parents via UTIL messages (line 2-3). When an agent $a_i$ receives a UTIL message from a child $a_c$, it joins its private functions w.r.t. its (pseudo) children in branch $a_c$ (line 5-6), and eliminates all the belonging variables whose highest (pseudo) parents are $a_i$ from the utility table (line 7). Here, $EV(a_i,a_c)$ is given by
$$EV(a_i,a_c)=PC(a_i)\cap Desc(a_c)\cup\{a_c\}\backslash ID(a_i)$$
Then $a_i$ joins the eliminated utility table with its running utility table $util_i$. It is worth mentioning that computing the set of elimination variables $EV(a_i,a_c)$ does not require agents to exchange their relative positions in a pseudo tree. Specifically, each variable is associated with a counter which is initially set to the number of its parent and pseudo parents. When its (pseudo) parent receives the UTIL message containing it, the counter decreases. And the variable is added to the set of elimination variables as soon as its counter equals zero.

After receiving all the UTIL messages from its children, $a_i$ propagates the utility table $util_i$ to its parent if it is not the root (line 12). Otherwise, the value propagation phase starts (line 10).
\begin{figure}
	\centering
	\includegraphics[scale=1]{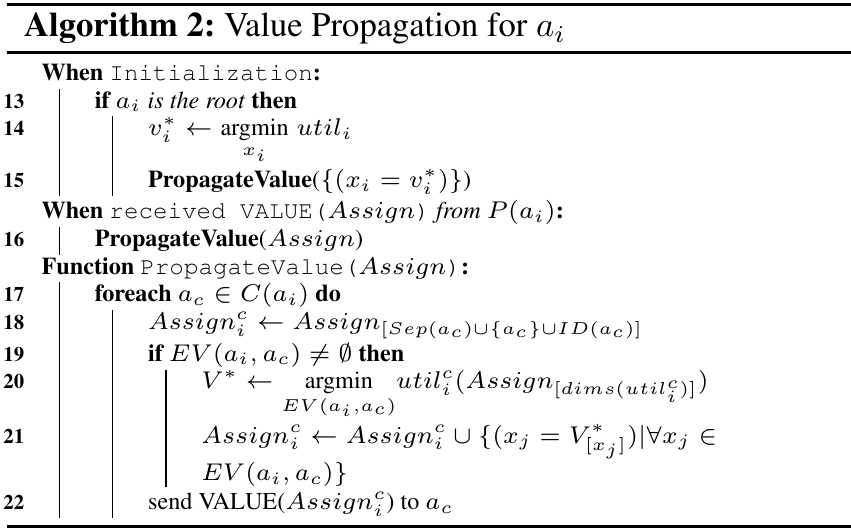}
	\label{vp}
\end{figure}
\begin{figure}
	\centering
	\includegraphics[scale=.4]{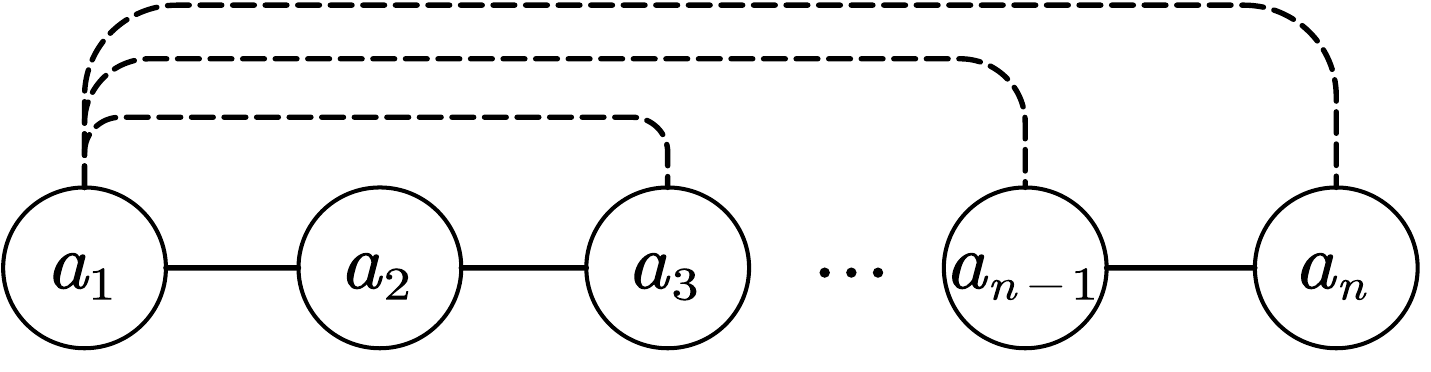}
	\caption{A chain-like pseudo tree}
	\label{chain}
\end{figure}
\subsection{Value Propagation Phase}
In contrary to the one in vanilla DPOP which determines the optimal assignment locally for each variable, the value assignment phase in AsymDPOP should be specialized to accommodate the non-local elimination. Specifically, since a variable is eliminated at its highest (pseudo) parent, the parent is responsible for selecting the optimal assignment for that variable. Thus, the value messages in our algorithm would contain not only the assignments for ancestors, but also assignments for descendants. Algorithm 2 presents the sketch of value propagation phase. 

The phase is initiated by the root agent selecting the optimal assignment (line 14). Given the determined assignments either from its parent (line 16) or computed locally (line 15), agent $a_i$ selects the optimal assignments for the eliminated variables in each branch $a_c\in C(a_i)$ by a joint optimization over them (line 17-20), and propagates the assignments together with the determined assignments to $a_c$ (line 21-22). The algorithm terminates when each leaf agent receives a VALUE message.  
\subsection{Complexity Analysis}
\newtheorem{theorem}{Theorem}
\begin{theorem}
	The size of a UTIL message produced by an agent is exponential in the number of its separator and its interface descendants.
\end{theorem}
\begin{proof}
	We prove the theorem by showing a UTIL message produced by agent $a_i$ contains the dimensions of $Sep(a_i)$ and $ID(a_i)\cup\{a_i\}$. The UTIL message must contain the dimensions of $ID(a_i)\cup\{a_i\}$ since $a_i$ is not the highest (pseudo) parent of $a_j,\forall a_j\in ID(a_i)\cup\{a_i\}$. On the other hand, according to the definition to $ID(a_i)$, the UTIL message cannot contain the dimensions of $Desc(a_i)\backslash ID(a_i)$ since for each $a_j\in Desc(a_i)\backslash ID(a_i)$ it must exist an agent $a_k\in Desc(a_i)\cup \{a_i\}$ such that $a_k$ is the highest (pseudo) parent of $a_j$ and thus the variable is eliminated at $a_k$ (line 7). Finally, the UTIL message contains $Sep(a_i)$ according to \cite{petcuF05} (Theorem 1). Thus, the size of the UTIL message is exponential in $|Sep(a_i)|+|ID(a_i)|+1$ and the theorem is concluded.

\end{proof}

\section{Tradeoffs}
As shown in Section 3.3, AsymDPOP suffers serious scalability problems in both memory and computation. In this section, we propose two tradeoffs which make AsymDPOP a practical algorithm. 
\subsection{Table-Set Propagation Scheme: A Tradeoff between Memory and Privacy}
The utility propagation phase of AsymDPOP could be problematic due to the unacceptable memory consumption when the pseudo tree is poor. Consider the pseudo tree shown in Fig. \ref{chain}. Since every agent is constrained with the root agent, according to the GNLE all the variables can only be eliminated at the root agent, which incurs a memory consumption of $O(d^{n})$ due to the join operation in each agent. Here, $d=\max_{i}|D_i|$. Besides, a large utility table would also incur unacceptable computation overheads due to the join operations and the elimination operation (line 5-7). 

We notice that utility tables are divisible before performing eliminations. Thus, instead of propagating a joint and high-dimension utility table to a parent agent, we propagate a set of small utility tables. In other words, we wish to reduce the unnecessary join operations (i.e., line 1 and line 7 in the case of $EV(a_i,a_c)=\emptyset$) which could cause a dimension increase during the utility propagation phase. On the other side, completely discarding join operations would incur privacy leakages. For example, if $a_n$ chooses to propagate both $f_{n,n-1}$ and $f_{n1}$ without performing the join operation to $a_{n-1}$, then $a_{n-1}$ would know the private functions of $a_n$ directly. Thus, we propose to compromise the memory consumption and the privacy by a parameter $k_p$ controlling the the maximal number of dimensions of each local utility table. We refer the tradeoff as a table-set propagation scheme (TSPS). 

Specifically, when an agent sends a UTIL message to its parent, it first joins its private functions w.r.t. its parents with the received utility tables whose dimensions contain the dimensions of these private functions. Notice that the first step does not incur a dimension increase, and can reduce the number of utility tables. Finally, it propagates the set of utility tables to its parent.

Consider the pseudo tree shown in Fig. \ref{chain} again. Assume that $k_p=3$ and then agent $a_n$ would propagate the utility set $util_n=\{f_{n,n-1}\otimes f_{n1}\}$ to $a_{n-1}$. Since there is no elimination in $a_{n-1}$, it is unnecessary to perform the join operation in line 7. Thus, $a_{n-1}$ would propagate the utility set $util_{n-1}=\{f_{n-1,n-2}\otimes f_{n-1,1},util_{n-1}^n\}$ to $a_{n-2}$. It can be concluded that TSPS in the example only requires $O(nd^3)$ space, which is much smaller than the one required by GNLE. Formally, we have the following theorem.
\begin{theorem}
	The size of each utility table of an agent $a_i$ in TSPS is no greater than 
	$$
	d^{\max(\min(|AP(a_i)|, k_p), \max\limits_{a_c\in C(a_i)} (|ID(a_i)\cap (ID(a_c)\cup \{a_c\})|+|Sep(a_c)|)}$$
\end{theorem}
\begin{proof}
	According to Theorem 1, the dimension of each utility table from child $a_c\in C(a_i)$ is a subset of $Sep(a_c)\cup ID(a_c)\cup\{a_c\}$. Since TSPS omits the join operation in line 7, the maximal size of received utility tables of $a_i$ is
	$$
	d^{\max\limits_{a_c\in C(a_i)} (|ID(a_i)\cap (ID(a_c)\cup \{a_c\})|+|Sep(a_c)|)}
	$$
	Since the local utility is partitioned into utility tables according to $k_p$, the maximal size of local utility tables of $a_i$ is 
	$$d^{\min(|AP(a_i)|,k_p)}$$
	Thus the theorem holds.
\end{proof}

\subsection{Mini-batch Elimination Scheme: A Tradeoff between Time and Space}
TSPS could factorize a big utility table to a set of smaller utility tables, which allows us to reduce the computational efforts when performing eliminations by a decrease-and-conquer strategy. Taking Fig. \ref{chain} as an example, to perform the elimination, $a_1$ in GNLE has no choice but to optimize a big utility table over variables $x_2,\dots x_n$ (line 7), which requires $O(d^n)$ operations. Instead, combining with TSPS ($k_p=2$) we could exploit the the structure of each small utility table by arranging the min operators among them to reduce computational complexity. That is, instead of performing
$$\min_{x_2,\dots x_n}u_{1,2,\dots,n}
$$
we perform
$$
\min_{x_2}(f_{12}+f_{21}+\cdots+\min_{x_n}(f_{1n}+f_{n1}+f_{n,n-1}+f_{n-1,n})\cdots)$$
which can be solved recursively from $x_n$ to $x_2$ and the overall complexity is $O(nd^3)$. In other words, we reduce the computational complexity by exploiting the independence among utility tables to avoid unnecessary traverses.
\begin{figure}
	\centering
	\includegraphics[scale=.45]{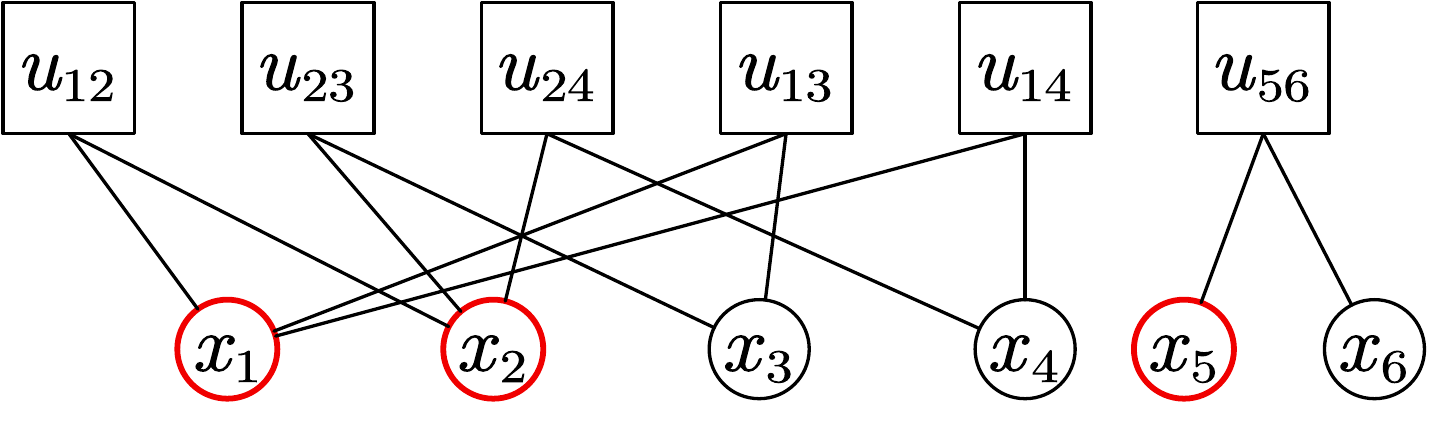}
	\caption{A set of utility tables}
	\label{factor-graph}
\end{figure}
 
However, completely distributing the min operators into every variable would incur high memory consumption as a min operator could implicitly join utility tables to a big and indivisible table. Although the problem can be alleviated by carefully arranging the min operators, it is not easy to find the optimal sequence of eliminations in practice. Consider the utility tables shown as a factor graph in Fig. \ref{factor-graph} where square nodes represent utility tables and circle nodes represent variables. And the red circles represent the variables to be eliminated. Obviously, no matter how to arrange the elimination sequence, a 3-ary utility table must appear when eliminating $x_1$ or $x_2$. Instead, if we jointly optimize both $x_1$ and $x_2$, the maximal number of dimensions are 2.

We thus overcome the problem by introducing a parameter $k_e$ which specifies the minimal number of variables optimized in a min operator (i.e., the size of a batch), and refer the tradeoff as a mini-batch elimination scheme (MBES). Specifically, when performing an elimination operation, we first divide elimination variables into several groups whose variables share at least a common utility table. For each variable group, we divide the variables into batches whose sizes are at least $k_e$ if it is possible. For each batch, we perform optimization to the functions that are related to the batch and replace these functions with the results. The process terminates when all the variable groups are exhausted. 

Note that dividing variables into disjoint variable groups in the first step is crucial since optimizing independent variables jointly is equivalent to optimizing them individually. Taking Fig. \ref{factor-graph} for example, if we set $k_e=2$ and let $x_2,x_5$ be a batch, a 4-ary utility table over $x_1,x_3,x_4$ and $x_6$ still appear even if $x_2$ and $x_5$ is jointly optimized. 

\begin{figure}
	\begin{minipage}{\linewidth}
		\centering
		\includegraphics[scale=.35]{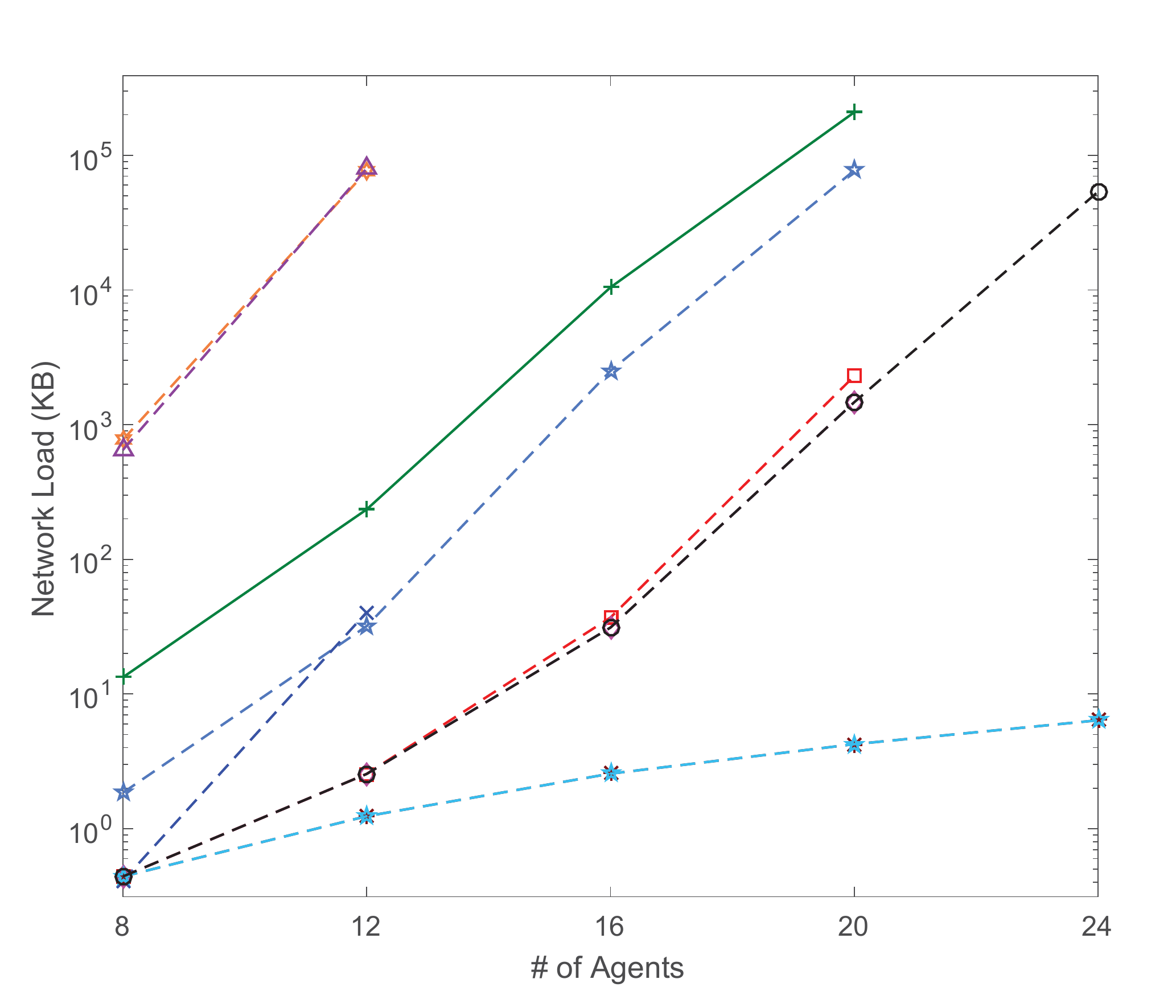}\\
		(a) network load
	\end{minipage}
	\begin{minipage}{\linewidth}
		\centering
		\includegraphics[scale=.35]{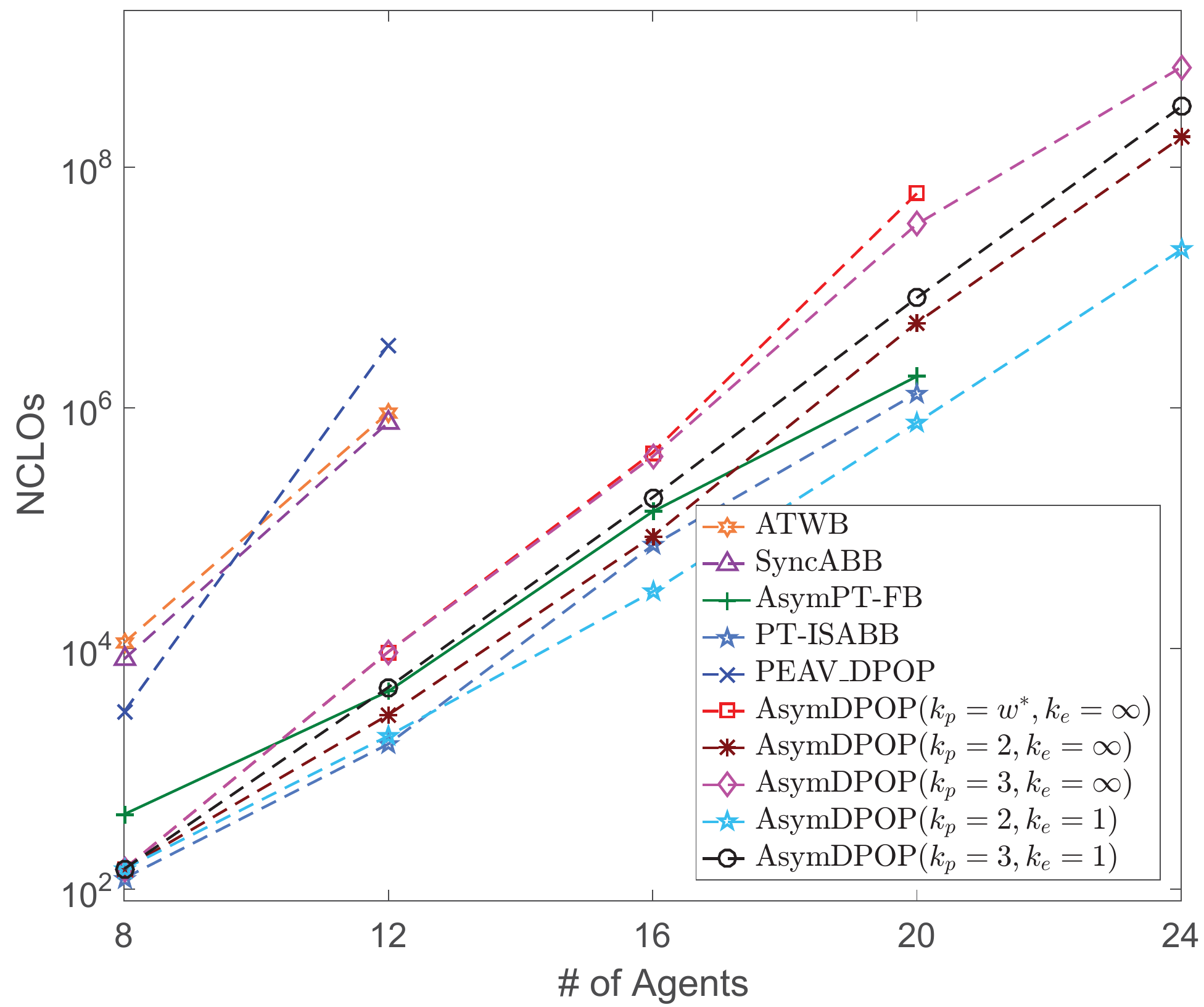}\\
		(b) NCLOs
	\end{minipage}
	\caption{Performance comparison under different agent numbers}
	\label{agent-number}
\end{figure}

\section{Experimental Results}
We empirically evaluate the performances of AsymDPOP and state-of-the-art complete algorithms for ADCOPs including PT-ISABB, AsymPT-FB, SyncABB, ATWB and the vanilla DPOP with PEAV formulation (PEAV\_DPOP) in terms of the number of basic operations, network load and privacy leakage. For inference-based algorithms, we consider the maximal number of dimensions during the utility propagation phase as an additional metric. Specifically, we use NCLOs \cite{Netzer2012} to measure the hardware-independent runtime in which the logic operations in inference-based algorithms are accesses to the utility tables while the ones in search-based algorithms are constraint checks. For the network load, we measure the size of information during an execution. Finally, we use entropy \cite{Brito2009Distributed} to quantify the privacy loss \cite{litov2017forward,grinshpounGZNM13}. For each experiment, we generate 50 random instances and report the medians as the results. 

\begin{figure}
	\begin{minipage}{\linewidth}
		\centering
		\includegraphics[scale=.35]{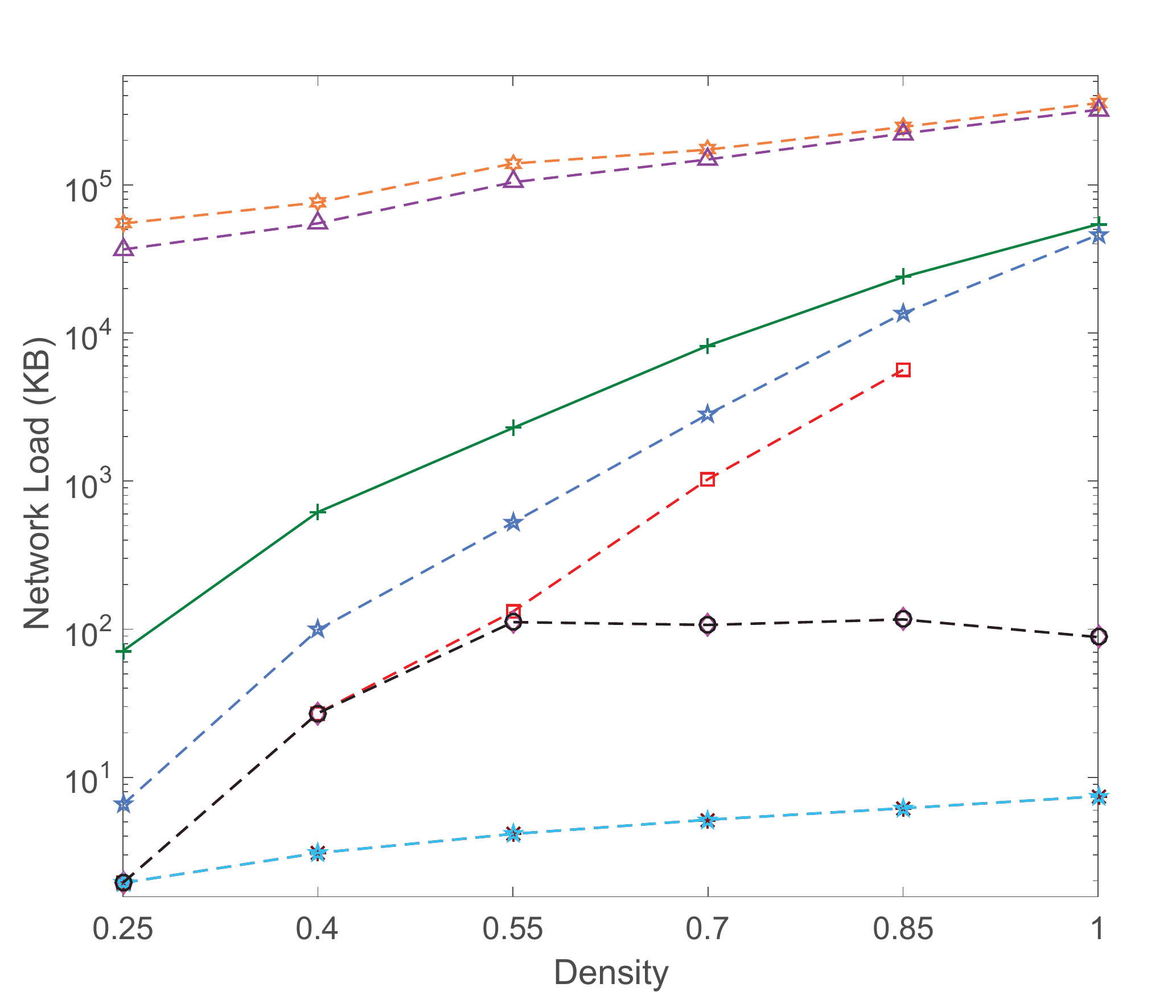}\\
		(a) network load
	\end{minipage}
	\begin{minipage}{\linewidth}
		\centering
		\includegraphics[scale=.35]{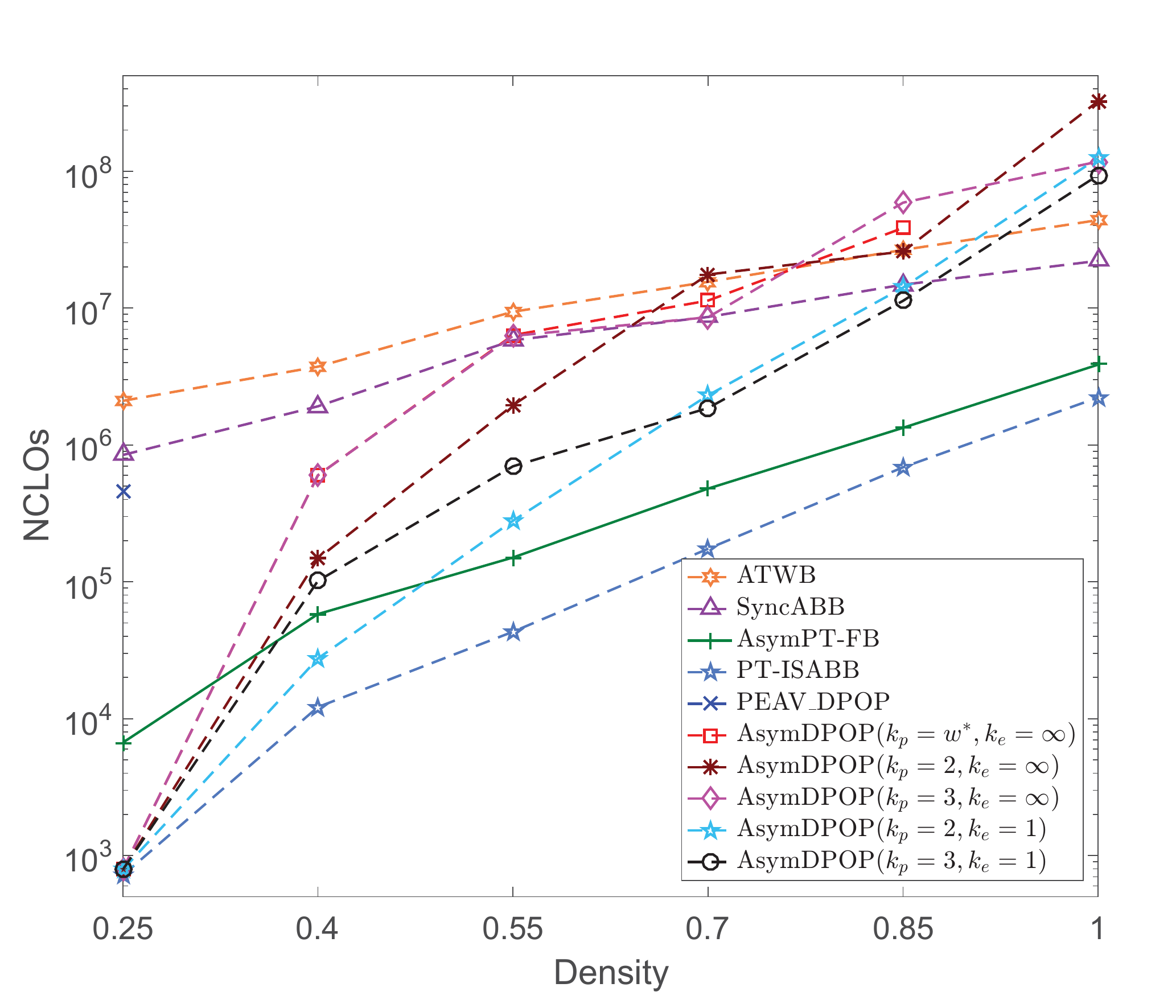}\\
		(b) NCLOs
	\end{minipage}
	\caption{Performance comparison under different densities}
	\label{densities}
\end{figure}
In our first experiment, we consider the ADCOPs with the domain size of 3, the density of 0.25 and the agent number varying from 8 to 24. Fig. \ref{agent-number} presents the experimental results under different agent numbers. It can be concluded that compared to the search-based solvers, our AsymDPOP algorithms exhibit great superiorities in terms of network load. That is due to the fact that search-based algorithms explicitly exhaust the search space by message-passing, which is quite expensive especially when the agent number is large. In contrast, our proposed algorithms incur few communication overheads since they follow an inference protocol and only require a linear number of messages. On the other hand, although PEAV\_DPOP also uses the inference protocol, it still suffers from a severe scalability problem and can only solve the problems with the agent number less than 12. The phenomenon is due to the mirror variables introduced by PEAV formulation, which significantly increases the complexity. More specifically, a UTIL message in PEAV\_DPOP contains the dimensions of mirror variables, which significantly increases the memory consumption.

\begin{figure}
	\begin{minipage}{\linewidth}
		\centering
		\includegraphics[scale=.35]{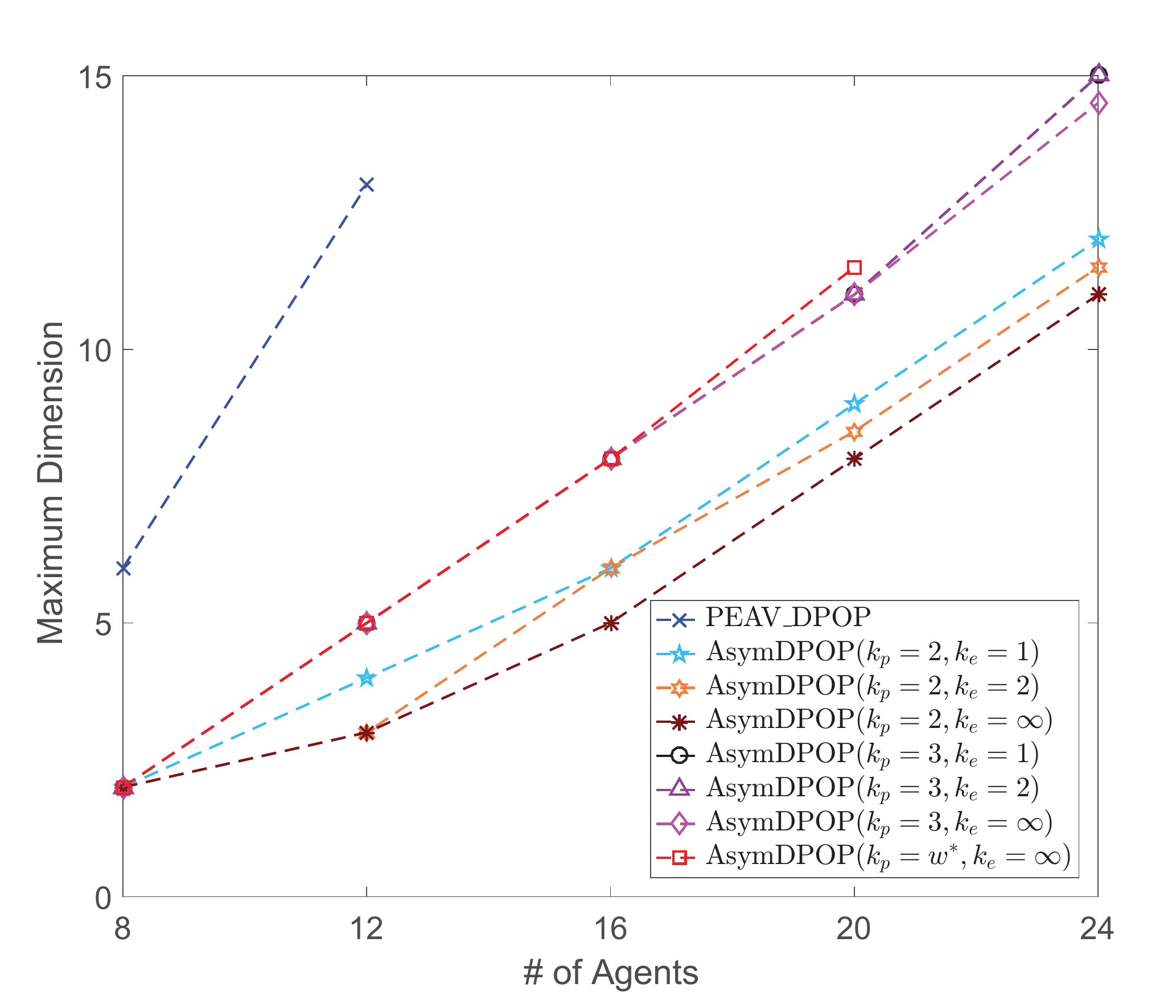}\\
		(a) maximum dimensions
	\end{minipage}
	\begin{minipage}{\linewidth}
		\centering
		\includegraphics[scale=.35]{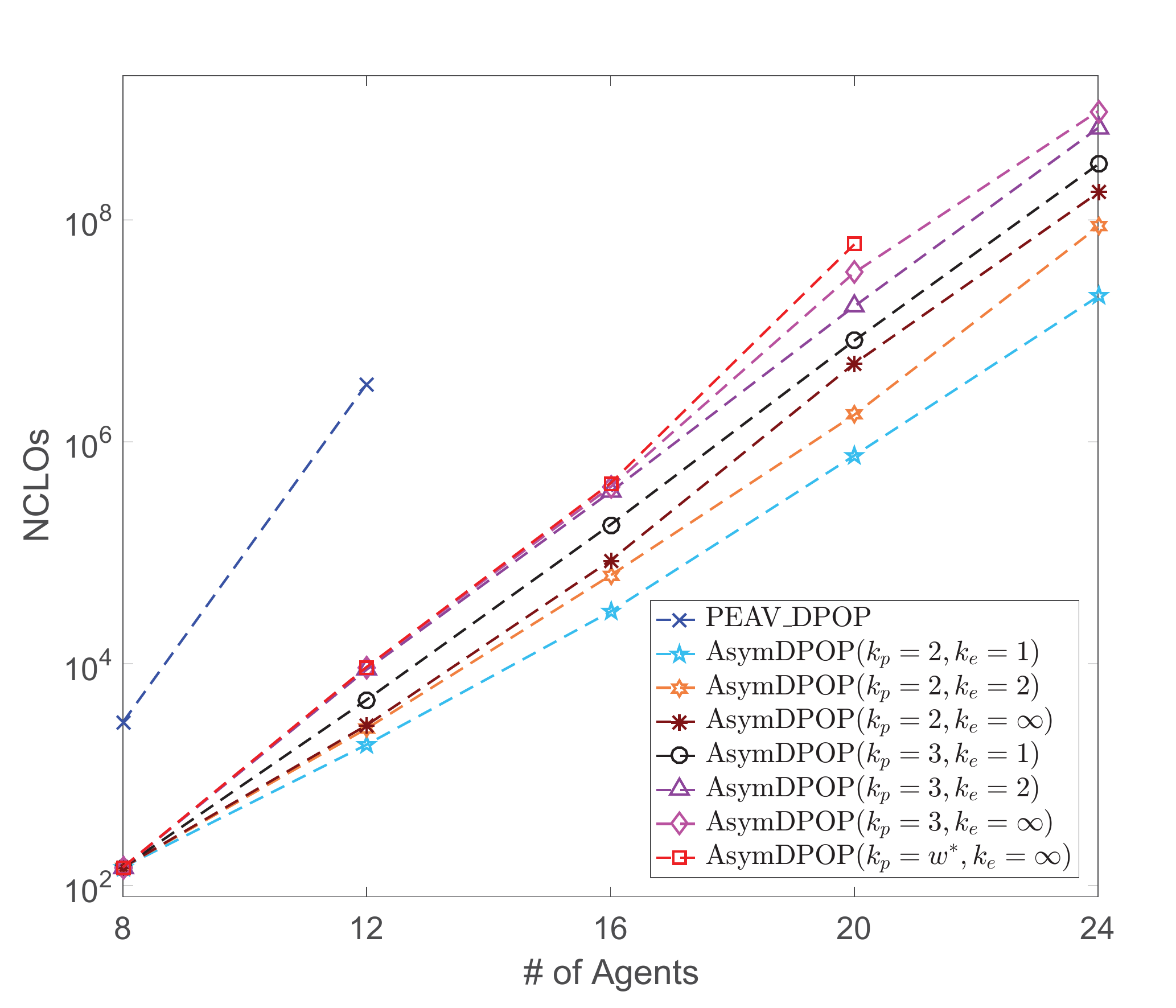}\\
		(b) NCLOs
	\end{minipage}
	\caption{Performance comparison of different batch sizes}
	\label{batches} 
\end{figure}
Fig. \ref{densities} presents the performance comparison under different densities. Specifically, we consider the ADCOPs with the agent number of 8, the domain size of 8 and the density varying from 0.25 to 1.  It can be concluded from Fig. \ref{densities}(a) that AsymDPOP with TSPS (i.e., $k_p<w^*$ where $w^*$ is the induced width) incurs significantly less communication overheads when solving dense problems, which demonstrates the merit of avoiding unnecessary join operations. Besides, it is interesting to find that compared to the one of AsymDPOP without TSPS (i.e., $k_p=w^*$), the network load of AsymDPOP$(k_p<w^*,\cdot)$ increases much slowly. That is due to the fact that as the density increases, eliminations are more likely to happen at the top of a pseudo tree. On the other hand, since unnecessary join operations are avoided in TSPS, eliminations are the major source of the dimension increase. As a result, agents propagate low dimension utility tables in most of the time. For NCLOs, search-based algorithms like AsymPT-FB and PT-ISABB outperform AsymDPOP algorithms when solving dense problems due to their effective pruning.

To demonstrate the effects of different batch sizes, we benchmark MBES with different $k_e$ when combining with TSPS with different $k_p$ on the configuration used in the first experiment. Specifically, we measure the maximal number of dimensions generated in the utility propagation phase, including the intermediate results of MBES. Fig. \ref{batches} presents the experimental results. It can be seen that AsymDPOP with a small batch size significantly reduces NCLOs but produces larger intermediate results, which indicates the necessity of the tradeoff. Besides, the performance of MBES significantly degenerates when combining with TSPS with a larger $k_p$. That is because the utility tables contain more dimensions in the scenario, and a utility table would be traversed more frequently when performing eliminations.

\begin{figure}
	\centering
	\includegraphics[scale=.35]{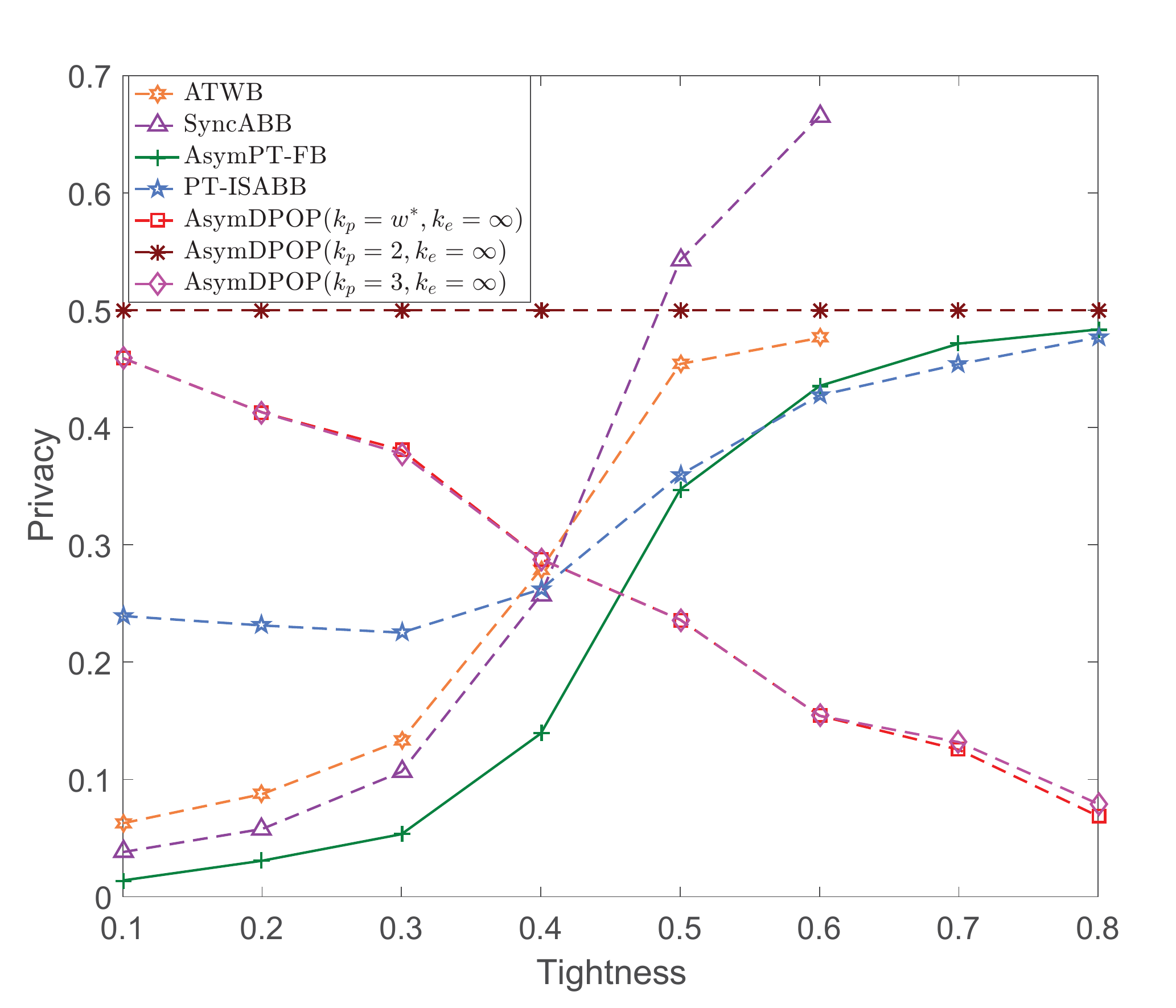}
	\caption{Privacy loss under different tightness}
	\label{privacy-loss}
\end{figure}
In the last experiment, we consider the privacy loss of different algorithms when solving asymmetric MaxDCSPs with different tightness. In particular, we consider the asymmetric MaxDCSPs with 10 agents, the domain size of 10, the density of 0.4 and the tightness varying from 0.1 to 0.8. Fig. \ref{privacy-loss} presents the results. It can be concluded from the figure that as the tightness grows, the search-based algorithms leak more privacy while the inference-based algorithms leaks less privacy. That is due to the fact that search-based algorithms rely on a direct disclosure mechanism to aggregate the private costs. Thus, search-based algorithms would leak more privacy when solving the problems with high tightness as they need to traverse more proportions of the search space. In contrast, inference-based algorithms accumulate utility through the pseudo tree, and an agent $a_i$ could infer the private costs of its (pseudo) child $a_c$ when the utility table involving both $x_i$ and $x_c$ is a binary table which is not a result of eliminations or contains zero entries.  Thus, AsymDPOP($k_p=2,\cdot$) leaks almost a half of privacy. On the other hand, since the number of prohibit combinations grows as the tightness increases, AsymDPOP($k_p\ge 3,\cdot$) incurs much lower privacy loss when solving high tightness problems.

\section{Conclusion} 
In this paper we present AsymDPOP, the first complete inference algorithm for ADCOPs. The algorithm incorporates three ingredients: generalized non-local elimination which facilitates the aggregation of utility in an asymmetric environment, table-set propagation scheme which reduces the memory consumption and mini-batch elimination scheme which reduces the operations in the utility propagation phase. We theoretically show its complexity and our empirical evaluation demonstrates its superiorities over the state-of-the-art, as well as the vanilla DPOP with PEAV formulation.

\section*{Acknowledgements}
We would like to thank the anonymous reviewers for their valuable comments and helpful suggestions. This work is supported by the Chongqing Research Program of Basic Research and Frontier Technology under Grant No.:cstc2017jcyjAX0030, Fundamental Research Funds for the Central Universities under Grant No.: 2018CDXYJSJ0026, National Natural Science Foundation of China under Grant No.: 51608070 and the Graduate Research and Innovation Foundation of Chongqing, China under Grant No.: CYS17023
\bibliographystyle{named}
\bibliography{ref}

\newpage
\appendix
\title{The Appendix} 
\maketitle
\renewcommand{\appendixname}{Appendix~\Alph{section}}
\renewcommand\thesection{\Alph{section}}
\section{AsymDPOP}

\subsection{An Example for AsymDPOP}

Fig.2 presents a pseudo tree. For better understanding, we take the agent $a_2$ to explain the concepts in a pseudo tree. Since $a_1$ is the only ancestor constrained with $a_2$ via a tree edge, we have $P(a_2)=\{a_1\}$, $PP(a_2)=\emptyset$ and $Sep(a_2)=\{a_1\}$.  Similarly, since $a_3$ and $a_4$ are the descendants constrained with $a_2$ via tree edge, we have $C(a_2)=\{a_3, a_4\}$, $PC(a_2)=\emptyset$, $Desc(a_2)=\{a_3, a_4\}$. Particularly, since $a_4\in Desc(a_2)$ is constrained with $a_1\in Sep(a_2)$, we have $ID(a_2)=\{a_4\}$.

AsymDPOP begins with the leaf agents ($a_3$ and $a_4$) sending their utility tables ($util_3$ and $util_4$) to their parent $a_2$, where $util_3 = f_{32}$ and $util_4 = f_{41}\otimes f_{42}$.

When $a_2$ receives the UTIL message from its child $a_3$ (assume  $a_3$'s UTIL message has arrived earlier), it stores the received utility table (i.e., $util_{2}^3 = util_3 = f_{32}$), and joins its private function $f_{23}$ to update the table $util_{2}^3$ (i.e., $util_{2}^3 = f_{32}\otimes f_{23}$). 
According to the definition of $EV(a_i, a_c)$ that all the belonging variables’ highest (pseudo) parent is $a_i$ in branch $a_c$. Here, $EV( a_2,a_3 )$ is given by
$$EV( a_2,a_3 ) =PC( a_2 ) \cap Desc( a_3 ) \cup \{ a_3 \} \backslash ID( a_2 )$$
$$=\emptyset \cap \emptyset \cup \{a_3\}\backslash \{a_4\}=\{a_3\}$$
Thus, $a_2$ eliminates variable $x_3$ from the utility table $util_{2}^3$. After that, $a_2$ joins the eliminated result to the running utility table $util_2$ ($util_2$ has been initialized to $f_{21}$). Similarly, upon receipt of the UTIL message from $a_4$, $a_2$ saves the received utility table (i.e., $util_2^4 = util_4 = f_{41}\otimes f_{42}$), then updates the utility table $util_2^4$ by joining the private function $f_{24}$ (i.e., $util_2^4 = f_{41}\otimes f_{42}\otimes f_{24}$). Since $EV( a_2,a_4 ) =\emptyset$, $a_2$ joins $util_2^4$ to the running utility table $util_2$ directly without any elimination. Since $a_2$ has received all the UTIL messages from its children, it propagates the utility table $util_2$ to its parent $a_1$. Here, we have 
$$util_2 = f_{21}\otimes(f_{41}\otimes f_{42}\otimes f_{24})\otimes (\underset{x_3}{\min}f_{23}\otimes f_{32})$$

When $a_1$ receives the UTIL message from $a_2$, it saves the received utility table (i.e., $util_1^2 = util_2$) and joins its private functions $f_{12}$ and $f_{14}$ to update $util_1^2$ (i.e., $util_1^2 = util_1^2 \otimes(f_{12}\otimes f_{14})$). After eliminating $x_2$ and $x_4$ ($EV\left( a_1,a_2 \right) =\{a_2,a_4\}$), $a_1$ joins the eliminated result into its running utility table $util_1$ ($util_1$ has been initialized to $null$). Here, we have
$$util_1 =
\underset{x_2,x_4}{\min}f_{12}\otimes f_{14}\otimes ( f_{21}\otimes ( f_{41}\otimes f_{42}\otimes f_{24} ) \otimes ( \underset{x_3}{\min}f_{23}\otimes f_{32} ) ) 
$$
Since $a_1$ is the root agent and receives all the UTIL messages from its children, it selects the optimal assignment $v_1^*$ for itself and the optimal assignments $v_2^*$ and $v_4^*$ for the eliminated variables $a_2$ and $a_4$. That is,
$$v_1^* = \underset{x_1}{\text{argmin}}\, util_1,$$
$$
( v_{2}^{*},v_{4}^{*} ) =\underset{x_2,x_4}{\text{argmin}}\, util_{1}^{2}( x_1=v_{1}^{*}  ) 
$$
Then $a_1$ propagates a VALUE message including the optimal assignment to its child $a_2$, where the optimal assignment is $\{ ( x_1=v_{1}^{*} ) ,( x_2=v_{2}^{*} ) ,( x_3=v_{3}^{*} ) \} $. 

When receiving the VALUE message from $a_1$, $a_2$ assigns the value $v_{2}^{*}$ for itself. Since $EV( a_2,a_3 ) =\{a_3\}$ and $EV( a_2,a_4 ) =\emptyset $, $a_2$ selects the optimal assignment for $a_3$ by performing optimization over $util_2^3$ with the determined assignment of $x_2$, that is 
$$
v_{3}^{*} =\underset{x_3}{\text{argmin}}( util_{2}^{3}( x_2=v_{2}^{*} ) ) 
$$

Then $a_2$ propagates the optimal assignment $\{(x_2 = v_{2}^{*}),(x_3 = v_{3}^{*})\}$ and $\{(x_1 = v_{1}^{*}),(x_2 = v_{2}^{*}),(x_4 = v_{4}^{*})\}$  to $a_3$ and $a_4$, respectively. Once $a_3$ and $a_4$ receives the VALUE messages, they assign for themselves. Since all leaf agents have received VALUE messages, the algorithm terminates.
\section{AsymDPOP with TSPS and MBES}
\subsection{Pseudo Code for AsymDPOP with TSPS and MBES}

\begin{figure}
	\begin{minipage}[t]{0.5\linewidth} 
		\includegraphics[scale=1]{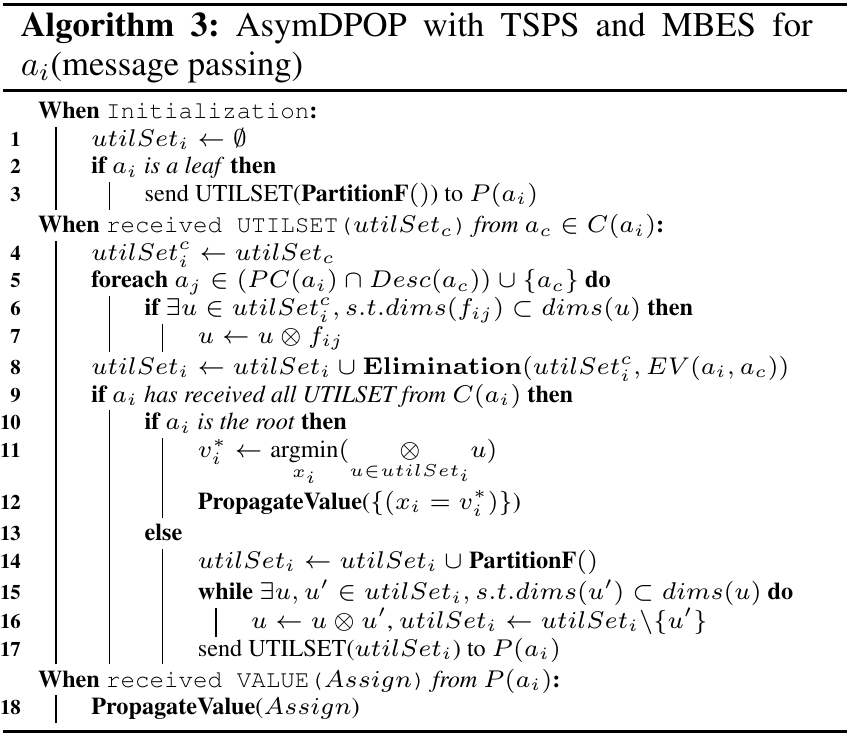} 
		\\
		\centering	
	\end{minipage}%
	\caption{The sketch of AsymDPOP with TSPS and MBES (message passing)}
\end{figure}

Fig.11 and Fig.12 present the sketch of AsymDPOP with TSPS and MBES which consists of two phases: utility set propagation phase and value propagation phase. Different from the utility propagation phase in AsymDPOP, the utility set propagation phase applies the TSPS to trade off between memory and privacy through propagating multiple small utility tables, and utilizes the MBES to trade off between time and space by sequential optimizations.

The utility set propagation phase begins with leaf agents sending their utility tables to their parents via UTILSET messages (line 2-3). Note that, the utility tables are obtained from the Function \textbf{PartitionF}. That is, leaf agents partition their local utilities into multiple smaller utility tables according to $k_p$ (line 32-41). If there is any residue, $a_i$ adds them to the utility tables (line 42-43). 

\begin{figure}
	\begin{minipage}[t]{0.5\linewidth} 
		\includegraphics[scale=1]{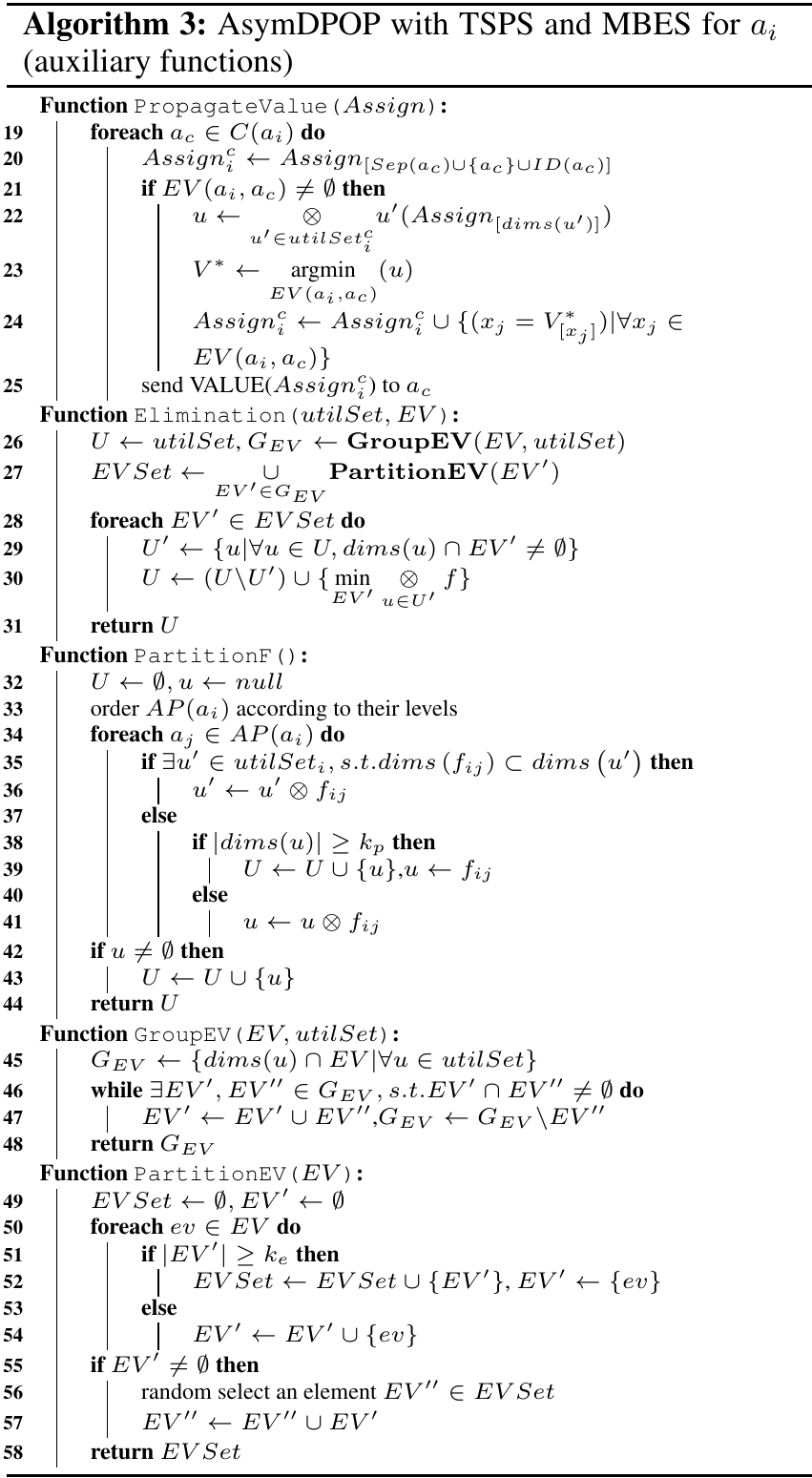} 
		\\
		\centering	
	\end{minipage}%
	\caption{The sketch of AsymDPOP with TSPS and MBES (auxiliary functions)}
\end{figure}
When receiving the UTILSET message from its child $a_c$, $a_i$ stores the received utility tables (i.e., $utilSet_i^c\gets utilSet_c$), then joins its private functions w.r.t. its (pseudo) children in branch $a_c$ into the utility tables $utilSet_i^c$ (line 4-7). Note that the join operation of its private functions w.r.t. its children does not increase the number of dimensions and should be applied accordingly to the related utility tables. Then the Function \textbf{Elimination} are used to implement sequential eliminations on the utility tables $utilSet_i^c$ with MBES (line 8, 26-30). In the Function \textbf{Elimination}, 
variables in $EV(a_i,a_c)$ are firstly divided into several groups ($G_{EV}$) such that the variables in each group share at least one common utility table (line 26). And then, for the variables in each group, $a_i$ partitions them into several sets (or batches) with the Function \textbf{PartitionEV} according to $k_e$ (line 27, 49-58). After obtaining the eliminated variable sets ($EVSet$), $a_i$ traverses the sets to optimize the utility tables ($utilSet$) (line 28-30). In detail, for each set, MBES performs optimization to the utility tables that are related to the eliminated variable in the set and replaces these utility tables with the optimized utility table.

When receiving all the UTILSET messages from its children, $a_i$ adds its local utility tables computed by the Function \textbf{FunctionF} into the optimized utility tables, then sends these utility tables to its parent if it is not the root agent (line 14-17). Otherwise, the value propagation phase starts (line 11-12). 

The value propagation phase is roughly the same as the one in AsymDPOP. The root agent $a_i$ selects the optimal assignments for itself and the eliminated variables ($EV(a_i,a_c)$) belong to each branch $a_c\in C(a_i)$. And since the utility set propagation uses the TSPS, there may have more than one utility table. So before selecting the optimal assignments for branch $a_c$ (line 23), $a_i$ joins all the tables in $utilSet_i^c$ with the assignments received from its parent (line 18) or computed locally (line 11-12) . After that, $a_i$ sends a VALUE message to its child $a_c$ (line 25). The algorithm terminates when each leaf agent receives a VALUE message.

\subsection{An Example of AsymDPOP with TSPS and MBES}
\begin{figure}
	\begin{minipage}[t]{0.5\linewidth} 
		\includegraphics[scale=0.5]{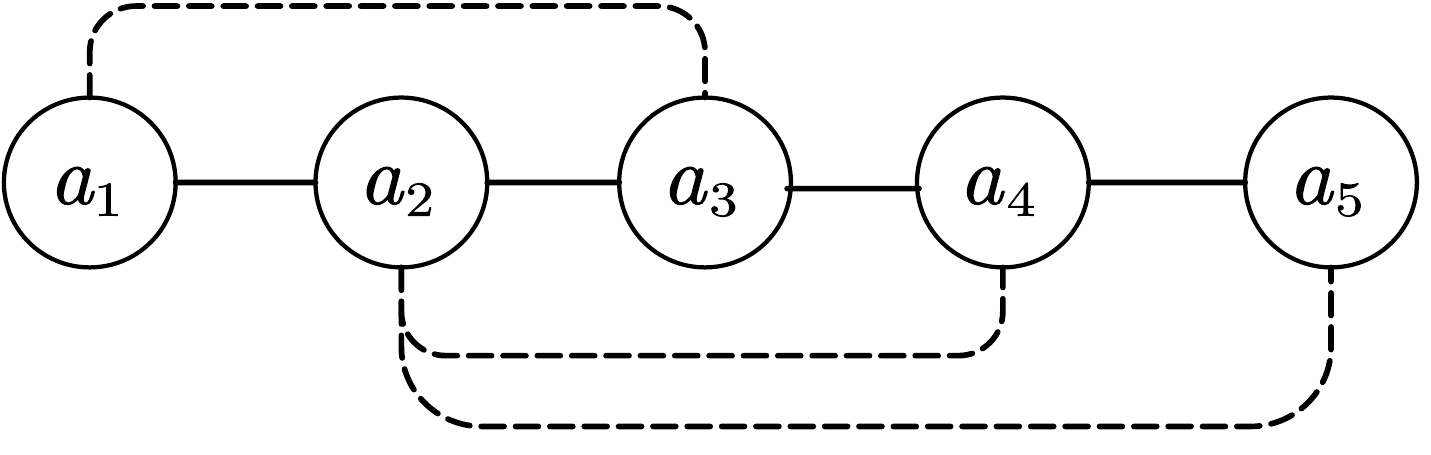} 
		\\
		\centering	
	\end{minipage}%
	\caption{A chain-like pseudo tree}
\end{figure}
For the convenience of further explanation, we denote the joint utility table $u_{ijk}$ as a function of three variables $x_i,x_j$ and $x_k$. And the index of the function is sorted in alphabetical order. 

We take Fig.13 as an example to demonstrate the Algorithm 3. And suppose that $k_p=3$ and $k_e=1$. The algorithm begins with the leaf agent $a_5$ by sending a UTILSET message with a utility table $u_{245}$ to its parent $a_4$. Here, $u_{245}$ is computed by Function \textbf{PartitionF}. Specifically, $a_5$ groups its private functions related with $AP(a_5)$ into a function set $\{f_{52},f_{54}\}$, since the dimension of its local utility table is three which satisfies $k_p=3$. Thus, $a_5$ obtains a 3-ary utility table $u_{245}$ by joining the functions in the function set (i.e., $u_{245}=f_{54}\otimes f_{52}$).

When $a_4$ receives the UTILSET message from its child $a_5$, it saves the utility table ($u_{245}$) and updates the utility table by joining its private functions $f_{45}$ and $f_{42}$ (i.e., $u_{245}=u_{245}\otimes (f_{45}\otimes f_{42})$). Since $a_4$ is not the highest (pseudo) parent of its any descendants ($EV(a_4,a_5)=\emptyset$), it does not need to perform elimination. Further more, since $a_4$ has received all the UTILSET messages from its children, it sends the utility tables $\{u_{245},u_{34}\}$ to  its parent $a_3$. Here, the utility table $u_{34}$ is computed by the Function \textbf{PartitionF} through dealing with the residual function $f_{43}$. 

Upon receipt of the UTILSET message from $a_4$, similarly, $a_3$ saves the utility tables and updates the utility table $u_{34}$ by joining its private functions $f_{34}$ (i.e., $u_{34}=u_{34}\otimes f_{34}$). Since $a_3$ is not the highest (pseudo) parent of $a_4$ or any other descendants, there is no elimination at $a_3$. Because the residual functions $f_{32}$ and $f_{31}$ could not be joined into the utility tables $\{u_{245},u_{34}\}$ without increase the number of dimensions, $a_3$ deals them with Function \textbf{PartitionF} and gets a utility table $u_{123}$. After that, it sends the utility tables $\{u_{123}, u_{245}, u_{34}\}$ to its parent $a_2$.

Once $a_2$ receives the UTILSET message from $a_3$, it also saves the utility tables firstly, and joins the relative private functions into the corresponding utility tables (i.e., $u_{245}=u_{245}\otimes (f_{25}\otimes f_{24}), u_{123}=u_{123}\otimes (f_{23}\otimes f_{21})$). Since $a_2$ is the highest (pseudo) parent of $a_4$ and $a_5$, that is $EV(a_2,a_3)=\{a_4,a_5\}$, the eliminations are preformed by the Function \textbf{Elimination}. In Function \textbf{Elimination}, $a_2$ firstly partitions the eliminated variables $x_4$ and $ x_5$ into a group, as they share a common utility table $f_{542}$. And since $k_e=1$, $a_2$ divides these variables into two batches and eliminates them from the utility tables $\{u_{123}, u_{245}, u_{34}\}$ one by one. After eliminating $x_5$ (assume $x_5$ is eliminated first), $a_2$ gets the utility tables $\{u_{123}, u_{24}, u_{34}\}$. Then $a_2$ obtains the new utility tables $\{ u_{123}, u_{2}, u_{3} \}$ by eliminating $x_4$ (i.e., $u_{2}=\underset{x_4}{\min}( u_{24}),u_{3}=\underset{x_4}{\min}( u_{34})$). After that, $a_2$ updates the utility table $u_{321}$ by a joint operation of the utility tables $\{u_{2}, u_{3}\}$ (i.e., $u_{123}=u_{123}\otimes u_{2} \otimes u_{3}$), and sends the updated utility table $\{u_{123}\}$ to its parent $a_1$.

When $a_1$ receives the UTILSET message from $a_2$, it saves the utility table, and then joins its private functions $f_{12}$ and $f_{13}$ into the utility table $u_{123}$ (i.e., $u_{123}=u_{123}\otimes f_{12}\otimes f_{13}$). Since $EV(a_1,a_2)=\{a_2,a_3\}$ and the variables $x_2$ and $x_3$ are both relative to the utility table $u_{321}$, they are partitioned into a single group. But since $k_e=1$, $a_1$ still eliminates the variables $x_2$ and $x_3$ from the utility table $\{u_{123}\}$, respectively, and obtains the eliminated utility table $\{u_{1}\}$. Since it is the root agent and has received all the UTILSET messages from its children, $a_1$ chooses the optimal value $v_1^*$ for itself based on the utility table $\{u_{1}\}$.

After that, the value propagation phase starts. $a_1$ selects the optimal values $v_2^*$ and $v_3^*$ for $a_2$ and $a_3$, then propagates the assignment $\{(x_1=v_1^*), (x_2=v_2^*), (x_3=v_3^*)\}$ to its child $a_2$. When receiving the VALUE message from $a_1$, $a_2$ selects the optimal assignments for $a_4$ and $a_5$, and sends the assignment $\{(x_1=v_1^*), (x_2=v_2^*), (x_3=v_3^*), (x_4=v_4^*), (x_5=v_5^*)\}$ to $a_3$. Upon receipt of the VALUE message from $a_2$, $a_3$ assigns for itself and sends the assignments received from $a_2$ to its child $a_4$. And $a_4$ performs just like $a_3$. The algorithm terminates when the leaf agent $a_5$ receives the VALUE message and assigns for itself. 

\section{Experiment Results}

\subsection{The Experimental Configuration}

\subsubsection{The Experiment with Different Agent Numbers}
\begin{itemize}
	\item Problem type: Random ADCOPs
	\item Agent numbers: [8, 24]
	\item Density: 0.25
	\item Domain size: 3
\end{itemize}

\subsubsection{The Experiment with Different Densities}
\begin{itemize}
	\item Problem type: Random ADCOPs
	\item Agent numbers: 8
	\item Density: [0.25, 1.0]
	\item Domain size: 8
\end{itemize}

\subsubsection{The Experiment with Different Domain size}
\begin{itemize}
	\item Problem type: Random ADCOPs
	\item Agent numbers: 8
	\item Density: 0.4
	\item Domain size: [4, 14]
\end{itemize}

\subsubsection{The Experiment with Different Tightness}
\begin{itemize}
	\item Problem type: ADCSPs
	\item Agent numbers: 10
	\item Density: 0.4
	\item Domain size: 10
	\item Tightness: [0.1, 0.8]
\end{itemize}

\subsection{The Experiment Results}

\begin{figure*}{
		\centering
		\subfloat[Network Load]{
			\includegraphics[scale=0.48]{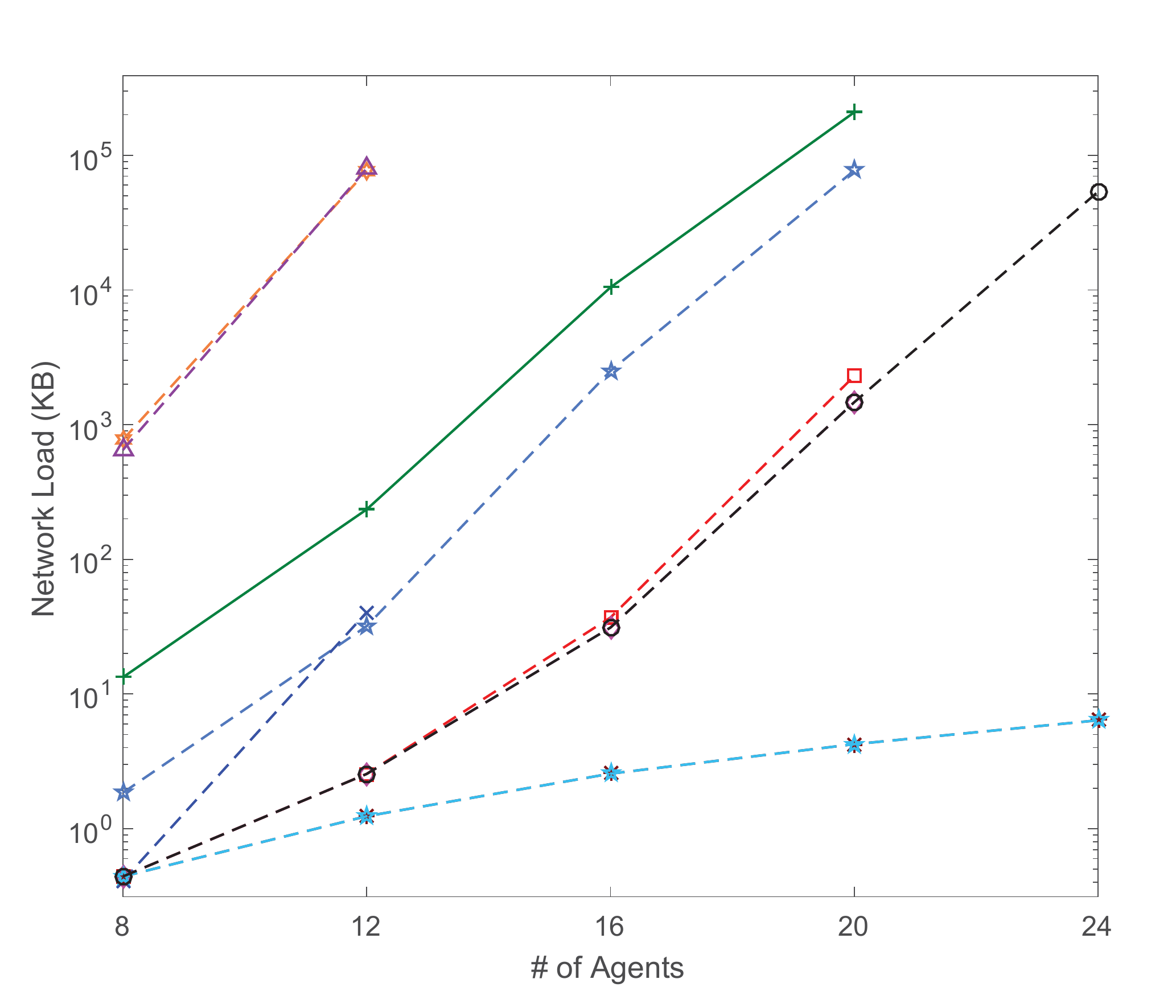} 
			\centering	
			
		}
		\subfloat[NCLOs]{
			\includegraphics[scale=0.48]{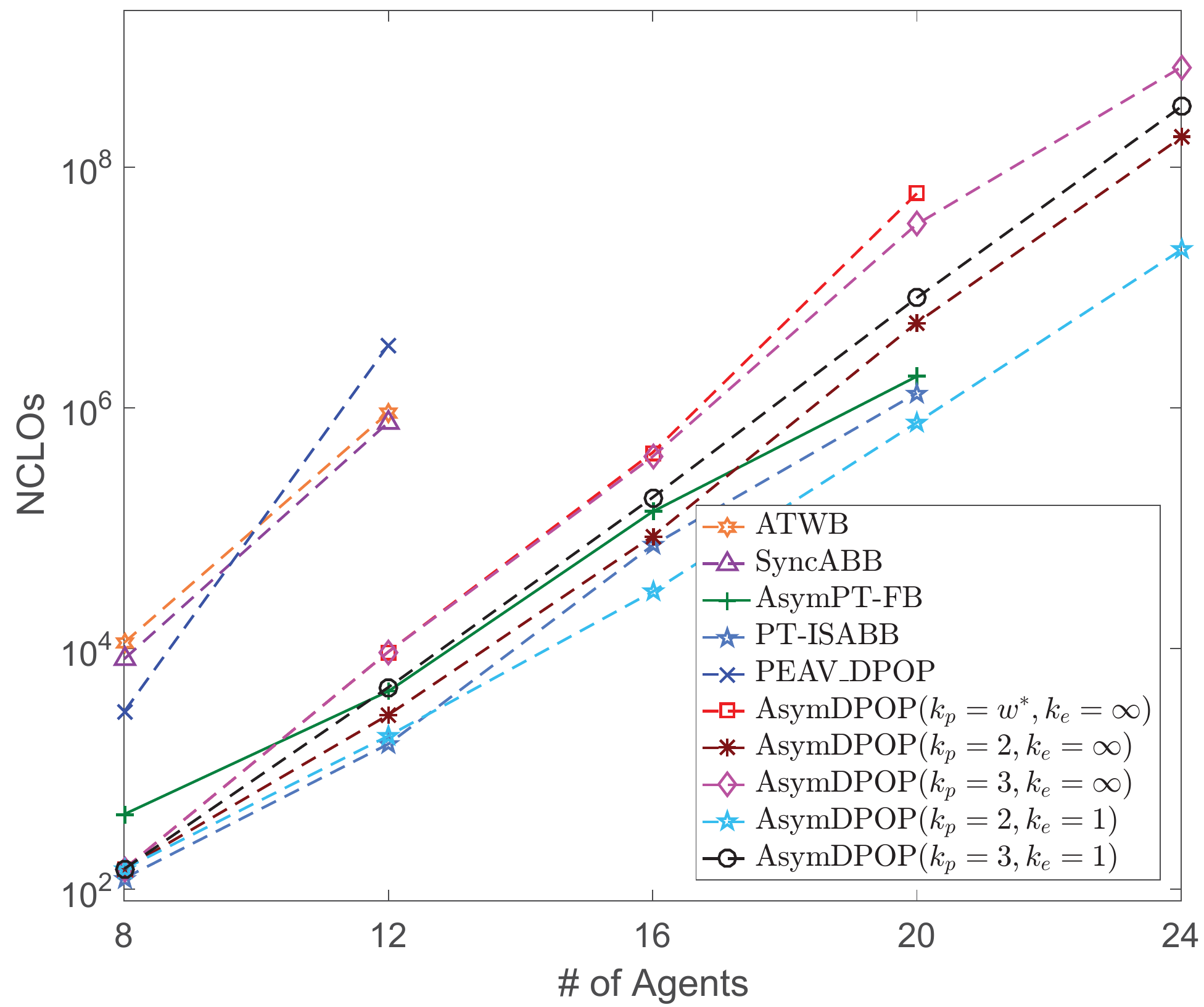} 
		}
		\caption{Performance comparison under different agent numbers}
	}
\end{figure*}

\begin{figure*}
	\centering
	\subfloat[Network Load]{
		\includegraphics[scale=0.48]{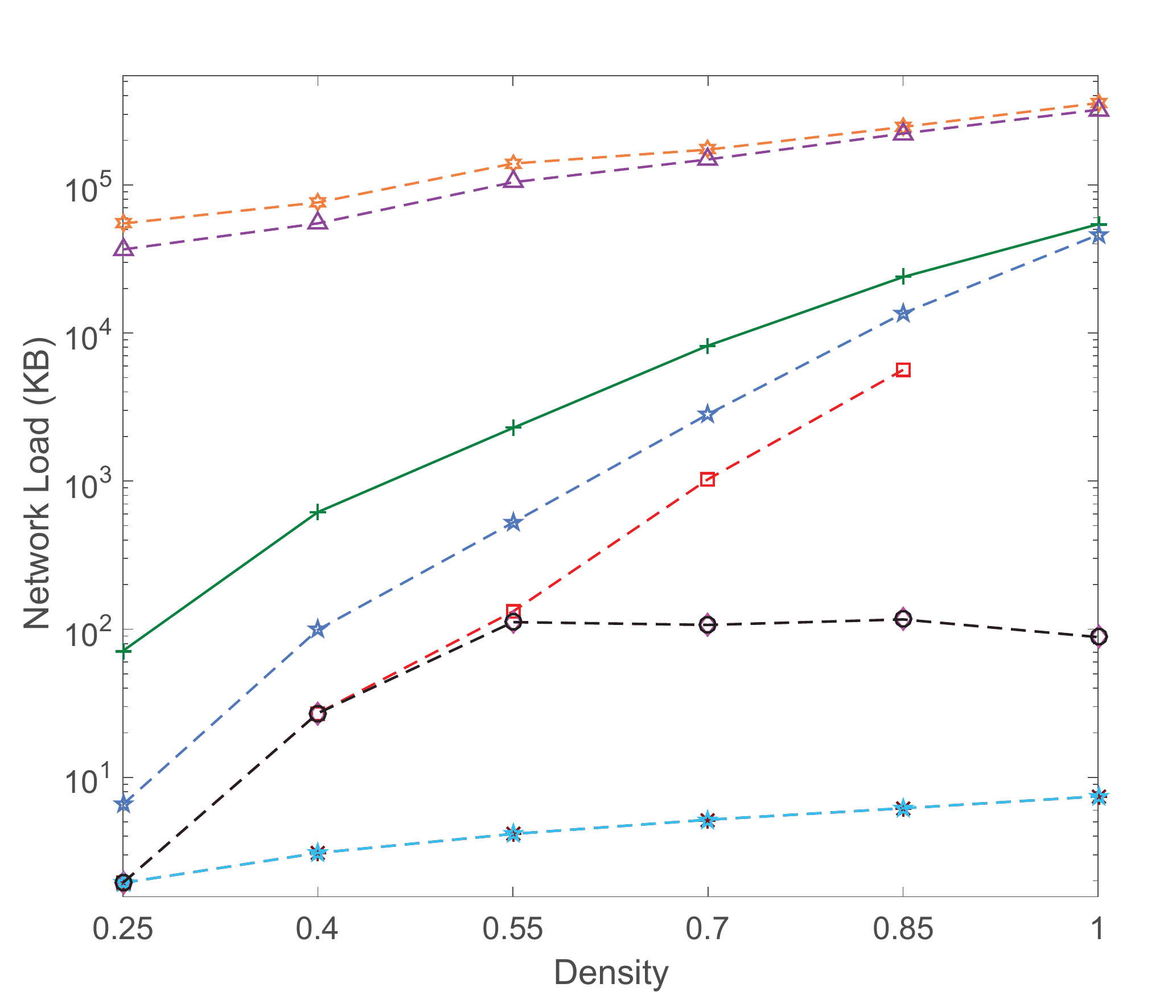} 
		
	}
	\subfloat[NCLOs]{
		\includegraphics[scale=0.48]{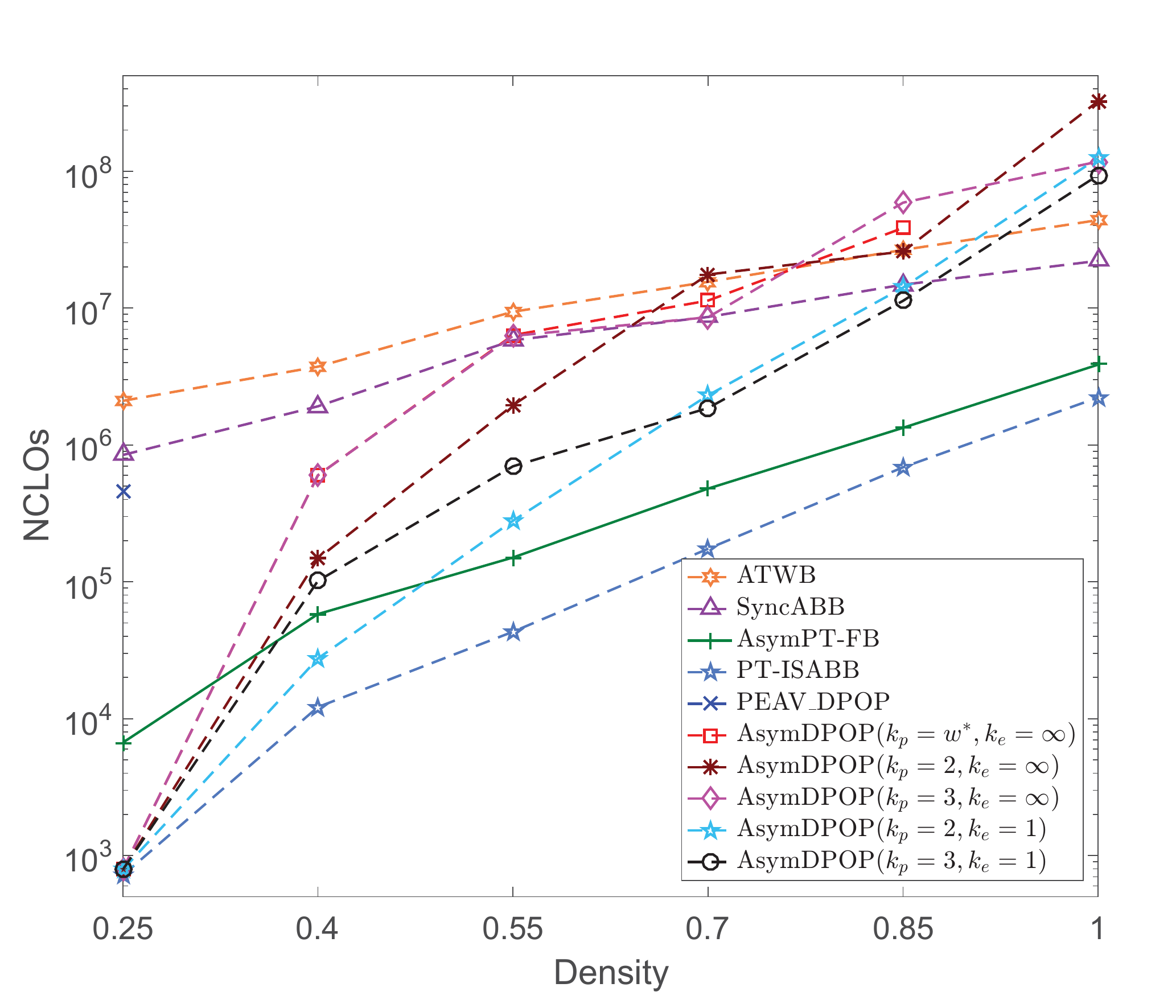} 
	}
	\caption{Performance comparison under different densities}
\end{figure*}
\begin{figure*}
	\centering
	\subfloat[Network Load]{
		\includegraphics[scale=0.5]{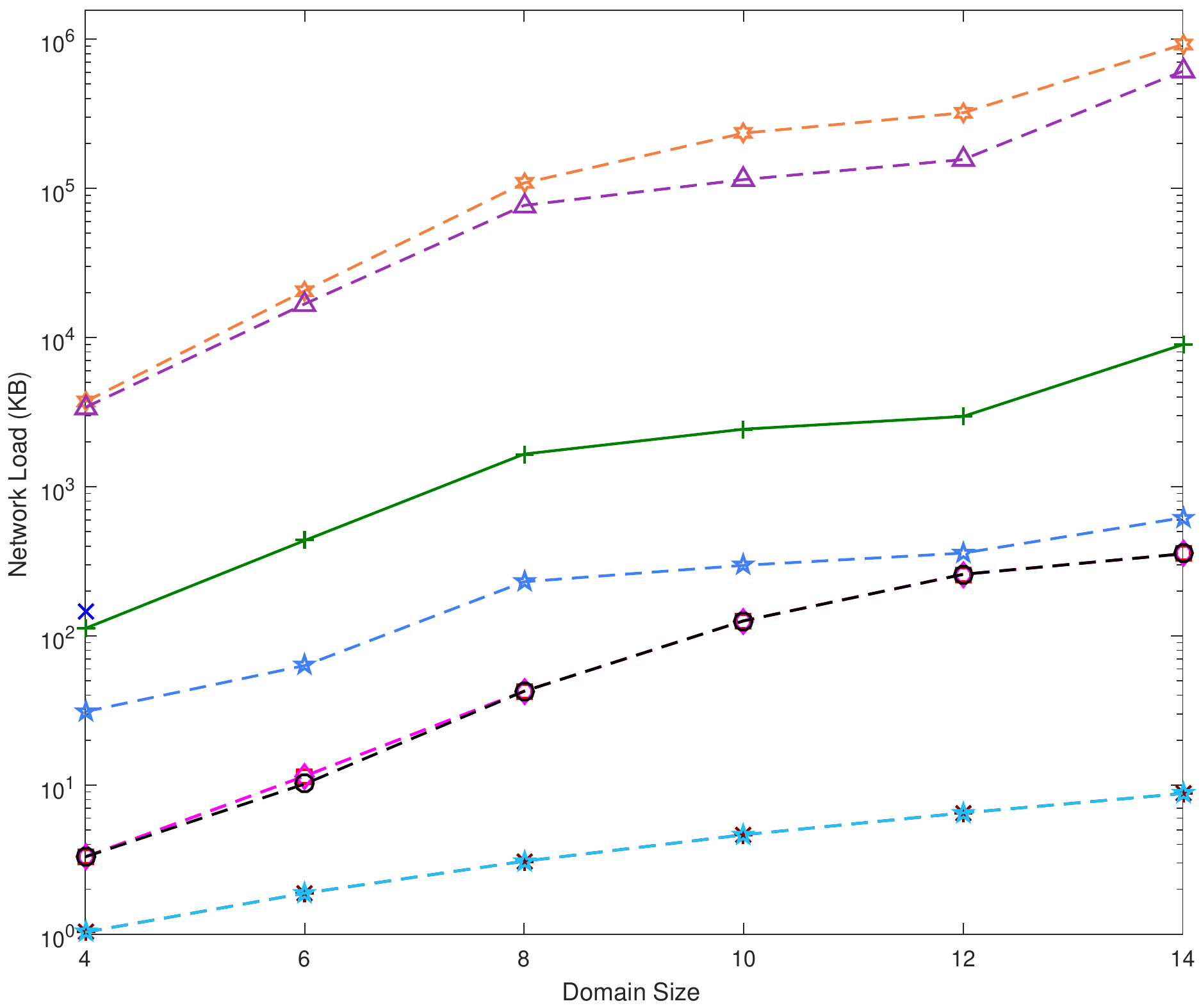} 
		\centering	
	}
	\subfloat[NCLOs]{
		\includegraphics[scale=0.5]{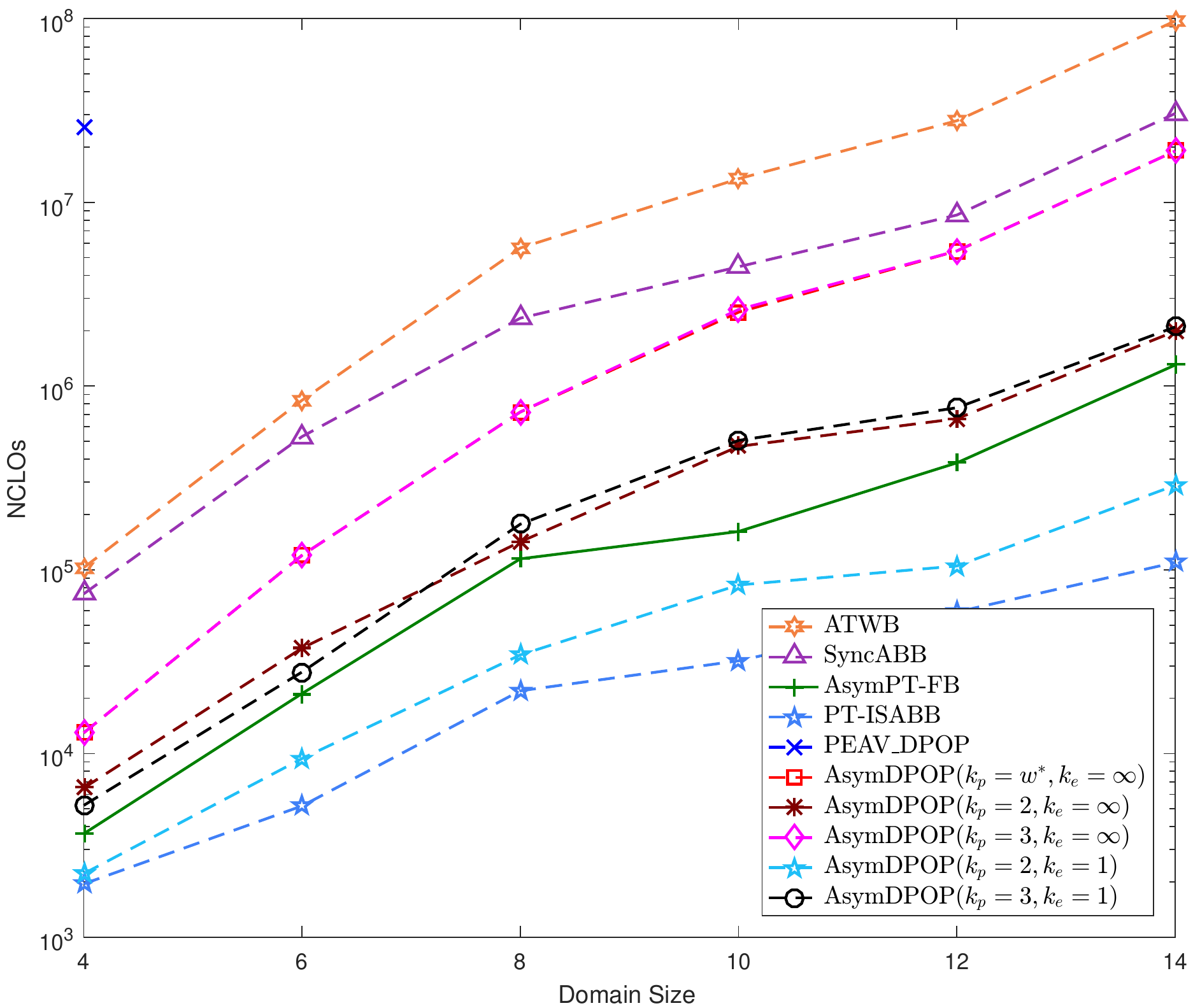} 
	}
	\caption{Performance comparison under different domain sizes}
\end{figure*}

\begin{figure*}
	\centering
	\subfloat[Maximum Dimension]{
		\includegraphics[scale=0.48]{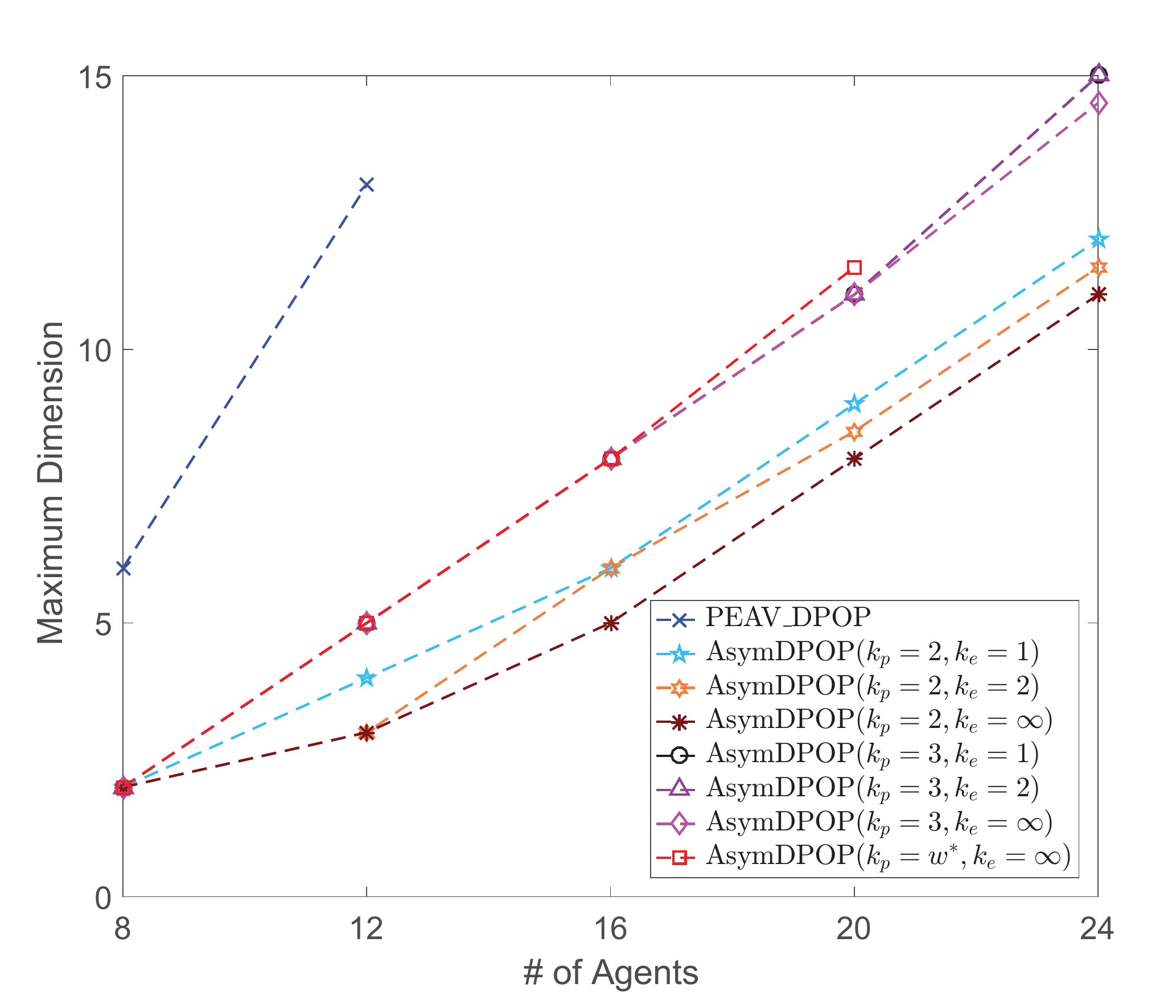} 
		\centering	
		
	}
	\subfloat[NCLOs]{
		\includegraphics[scale=0.48]{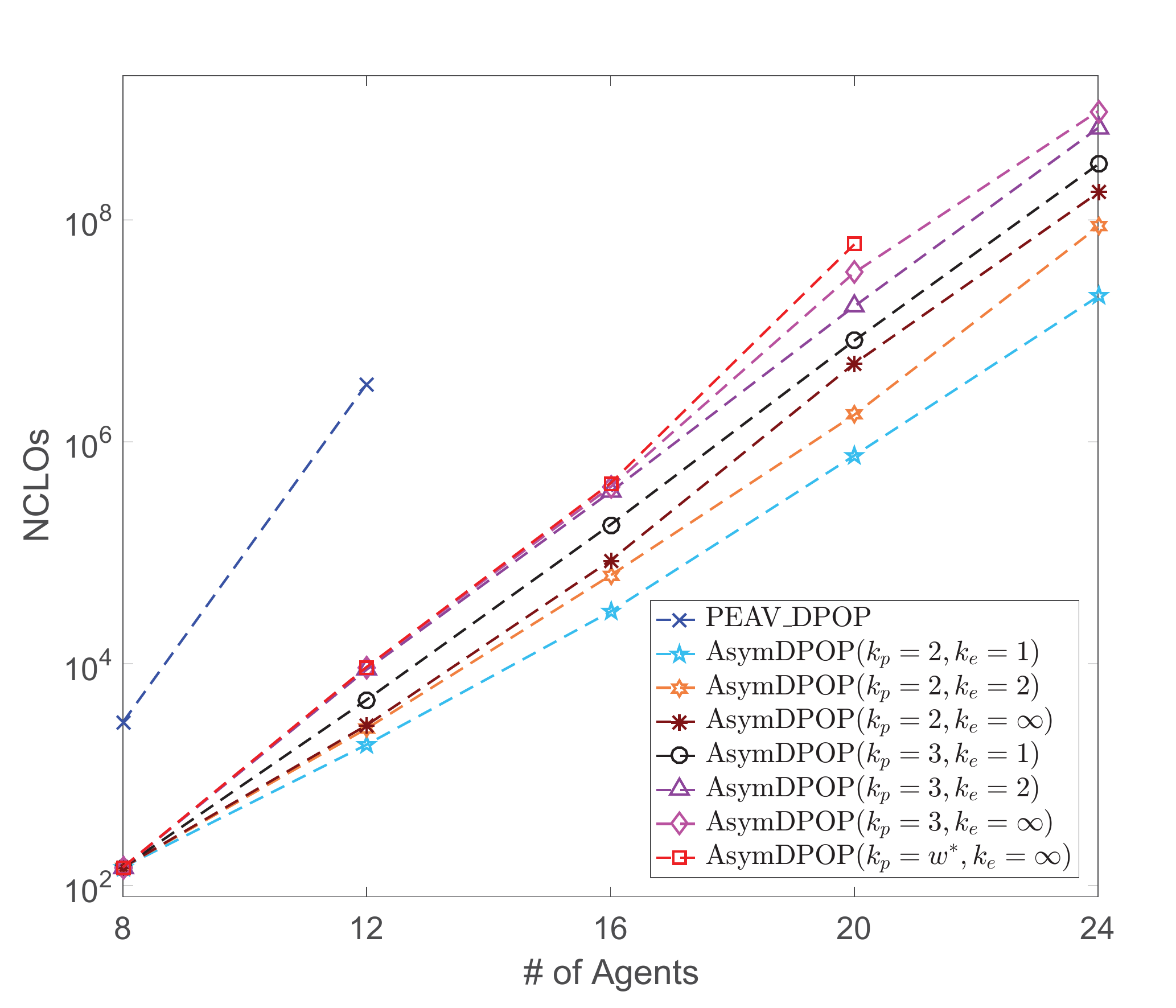} 
		
	}
	\caption{Performance comparison under different batch sizes}
\end{figure*}

\begin{figure*}
	\centering
	\subfloat[Network Load]{
		\includegraphics[scale=0.48]{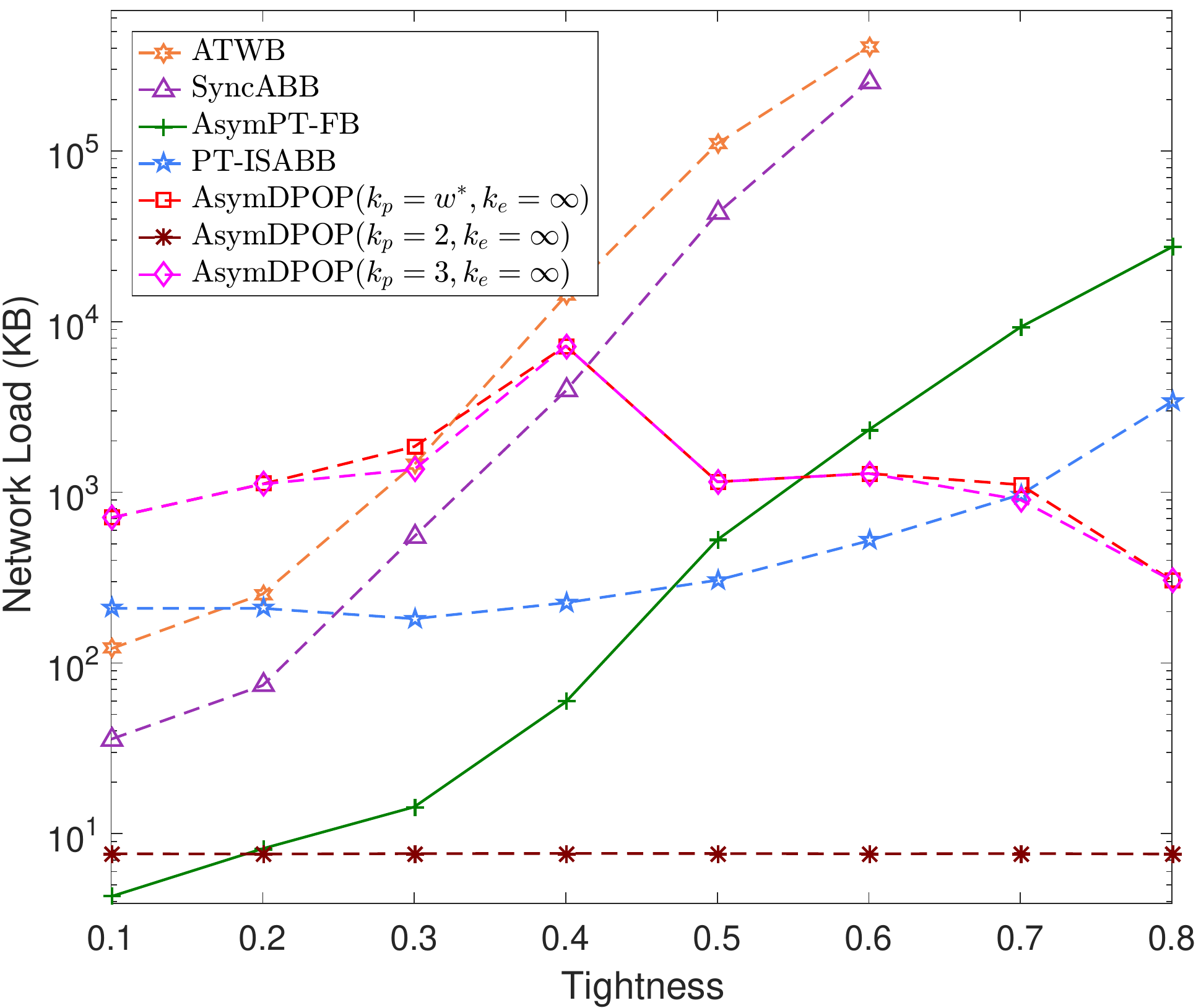} 
		\centering	
	}
	\subfloat[NCLOs]{
		\includegraphics[scale=0.48]{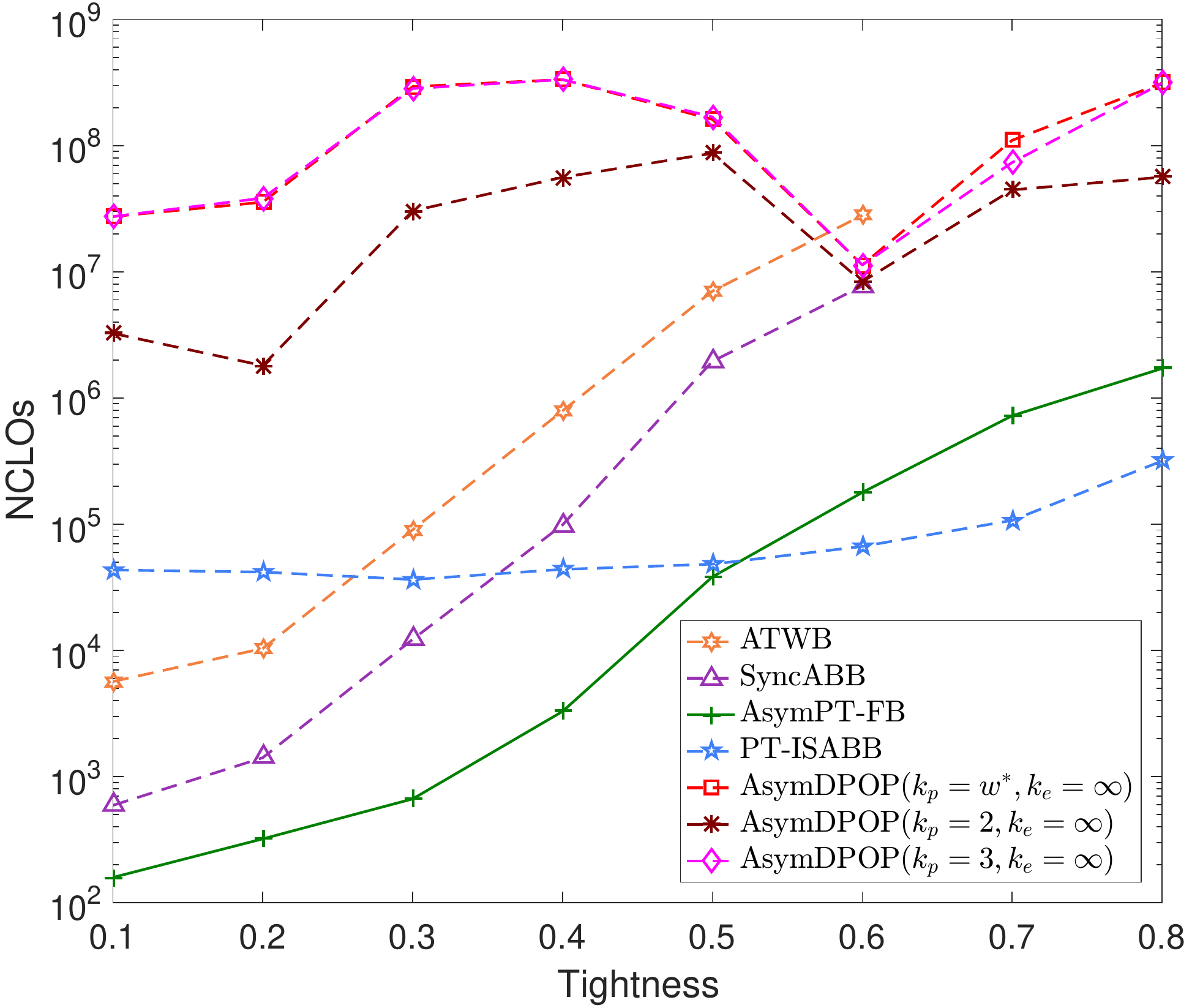} 
		\centering	
	}\\
	\subfloat[Privacy]{
		\includegraphics[scale=0.62]{dcsp_privacy.pdf} 
		
	}
	\caption{Performance comparison under different tightness}
\end{figure*}

\end{document}